\documentclass[fleqn,11pt]{wlscirep}
\usepackage[utf8]{inputenc}
\usepackage[T1]{fontenc}
\usepackage{hyperref}
\usepackage{authblk}
\usepackage{xcolor}
\usepackage{siunitx}
\usepackage{ragged2e}
\usepackage{algorithm}
\usepackage{algpseudocode}
\usepackage{amsthm}
\usepackage{tikz}
\usetikzlibrary{arrows.meta,calc,positioning,decorations.pathreplacing}
\usepackage{mathtools}

\newcommand{\name}{ATLAS}

\definecolor{chgnew}{rgb}{0,0.55,0}

\definecolor{R0}{rgb}{0,0,255}

\title{ATLAS: A Foundation Neural Sampler for Amorphous Materials}

\author[1,2,3,*,$\#$]{Mouyang Cheng}
\author[1,4,*,$\#$]{Denis Blessing}
\author[1,5,$\#$]{Botao Yu}
\author[4]{Gerhard Neumann}
\author[2,6]{Mingda Li}
\author[1,$\dagger$]{Carles Domingo-Enrich}
\author[1,$\dagger$]{Yuanqi Du}

\affil[1]{Microsoft Research New England, Cambridge, MA 02142, USA}
\affil[2]{Center for Computational Science and Engineering, MIT, Cambridge, MA 02139, USA}
\affil[3]{Department of Materials Science and Engineering, MIT, Cambridge, MA 02139, USA}
\affil[4]{Karlsruhe Institute of Technology, 76131 Karlsruhe, Germany}
\affil[5]{Department of Computer Science and Engineering, OSU, Columbus, OH 43210, USA}
\affil[6]{Department of Nuclear Science and Engineering, MIT, Cambridge, MA 02139, USA}

\affil[*]{These authors contributed equally.}
\affil[$\#$]{Work done during internship at Microsoft Research.}
\affil[$\dagger$]{Correspondence. Email: carlesd@microsoft.com, yuanqidu@microsoft.com}

\begin{abstract}
Amorphous materials exhibit exceptional mechanical and functional properties, yet their rugged energy landscapes are notoriously difficult to sample. Below the glass-transition temperature, conventional molecular dynamics and Monte Carlo become inefficient because equilibration relies on rare barrier-crossing events, while data-driven generative models are constrained by scarce and biased reference ensembles. 
Here, we introduce ATLAS, an efficient sampler that learns a diffusion process to generate Boltzmann-distributed amorphous structures directly from a target energy function. Parameterized by an equivariant graph neural network, ATLAS generalizes across system size, temperature, and composition. By exploiting the time reversal of the diffusion process, it enables efficient estimation of thermodynamic quantities and steering toward target observables. 
In two-dimensional Kob-Andersen systems, ATLAS reproduces parallel tempering Markov chain Monte Carlo structural distributions, free energies and entropies, achieving below 0.2\% free energy error in the low-temperature glass regime with over 500-fold fewer energy evaluations. 
In Cu-Zr and Cr-Co-Ni metallic glasses, ATLAS recovers experimentally observed short-range-order trends and steers structures toward prescribed order parameters and optimized bulk moduli. 
Moreover, composition-amortized pretraining outperforms composition-specific training from scratch, reduces inverse-design costs by several hundred-fold, and enables sampling with expensive universal machine learning interatomic potentials. Coupled to a large language model agent, ATLAS searches an eight-element space for high-entropy metallic glasses balancing stiffness and ductility, identifying a converged Pareto frontier within 480 oracle evaluations.
Together, these results establish ATLAS as a foundation model for sampling, steering and designing amorphous materials.
\end{abstract}

\begin{document}

\flushbottom
\maketitle
\thispagestyle{empty}

\section*{Introduction}
Amorphous solids have emerged as a frontier in materials science and condensed matter physics \cite{mauro2014two,liu2025amorphous}. From bulk oxide glasses \cite{zachariasen1932atomic} and metallic glasses \cite{cheng2011atomic} to two-dimensional amorphous materials \cite{shi2025emergence}, these non-crystalline systems can transform structural disorder into tunable mechanical, transport, optical and chemical properties, sometimes beyond those of their crystalline counterparts \cite{smith2013photochemical,hong2020ultralow,tian2023disorder,cai2026ceramic}.
These physical properties are controlled by a thermodynamic ensemble of metastable atomic configurations with short- and medium-range order, defects and rare structural rearrangements. 
Predicting these ensembles is therefore central to understanding structure-property relationships in disordered systems, and to enabling rational design and discovery of amorphous materials.

However, efficiently simulating amorphous materials remains a longstanding challenge because their potential energy surfaces (PESs) contain a vast number of structurally distinct metastable basins. At low temperatures $T$, transitions between these basins require rare collective rearrangements over barriers many times larger than $k_{\mathrm B}T$, making conventional simulations exceedingly slow \cite{ninarello2017models,berthier2023modern,leoni2026computational}.
Melt-quench protocols in molecular dynamics (MD) can generate realistic amorphous structures \cite{vollmayr1996cooling}, but their cooling rates are many orders of magnitude faster than those accessible in experiments, making the resulting configurations strongly protocol-dependent and not guaranteed to represent the target ensemble with correct Boltzmann weights.
Replica-exchange and swap Monte Carlo methods can improve equilibration in modeling glass-forming systems \cite{yamamoto2000replica,gazzillo1989equation,grigera2001fast}, but convergence, transferability and scalability remain challenging in low-temperature amorphous regimes.
Recent deep generative models are naturally suited to alleviating the challenge of PES sampling because they can efficiently represent complex, high-dimensional probability distributions over atomic configurations \cite{cheng2026artificial}.
For example, in the same spirit as recent advances in crystal structure prediction and \textit{de novo} generation \cite{xie2021crystal,jiao2023crystal,zeni2025generative}, diffusion- and flow-matching-based methods have been applied to generate diverse periodic amorphous structures \cite{yang2025generative,grenioux2025riemannian}, with extensions to conditional generation guided by target properties or structural constraints \cite{li2025conditional}. 
However, these approaches, like their crystalline counterparts, rely on structural training datasets. This requirement is especially limiting for amorphous materials, where no off-the-shelf database provides broad, equilibrated ensembles of amorphous configurations. 
Training data must therefore be generated by conventional simulations, which already introduces deviations from the target thermodynamic distribution. As a result, a generative model trained on these data may therefore reproduce these biases rather than recover the desired ensemble.

The limitations of training on precomputed data have motivated recent advances in establishing probabilistic models to directly sample high-dimensional Boltzmann distributions, whose unnormalized densities are specified by known energy functions, without relying on available structural data, following the idea of Boltzmann generators \cite{noe2019boltzmann}. 
Early approaches used normalizing flow-based models to learn an invertible map from a simple prior to the target Boltzmann distribution over continuous particle coordinates\cite{noe2019boltzmann,gabrie2022adaptive}; more recently, diffusion-based neural samplers have emerged as a more expressive and flexible alternative \cite{zhang2021path,vargasdenoising2023,richter2024improved,vargas2024transport}. 
Despite their theoretical foundation, neural samplers face significant practical challenges, ranging from limited sample efficiency and mode collapse to inaccurate mode weights \cite{he2025no}. These challenges have motivated the subsequent development of sample-efficient, mode-seeking neural samplers based on scalable fixed-point matching methods \cite{havens2025adjoint,liu2026adjoint,blessing2026bridge}. However, direct sampling of amorphous matter, or even realistic molecular systems more broadly, with neural samplers, remains challenging.

In this work, we introduce ATLAS, a neural sampler for amorphous matter that samples Boltzmann-weighted structural ensembles and can be flexibly applied to a broad range of downstream tasks.
ATLAS learns to transform a uniform prior under periodic boundary conditions into Boltzmann-distributed atomic configurations through a diffusion process.
An equivariant graph neural network (GNN) is trained on bootstrapped configurations generated by ATLAS itself to recover both the forward and backward stochastic dynamics, using the target potential energy $U(\mathbf{x})$ as direct supervision rather than reference structures.
Once trained, ATLAS generates independent amorphous configurations together with free energy estimates and Boltzmann-weighted observables, and can be generalized across chemical composition, temperature and system size. 
Benchmarks on two-dimensional Kob-Andersen (KA) glass formers against parallel tempering Markov chain Monte Carlo (PT-MCMC) show that ATLAS reproduces structural and thermodynamic observables, achieving below 0.2\% free energy error in the low-temperature glass regime ($T=0.2$) with more than 500-fold fewer potential energy evaluations.
Moreover, ATLAS can be efficiently steered toward target expectations of experimental observables, either through fine-tuning or inference-time control. We demonstrate this capability in binary Cu-Zr and ternary Cr-Co-Ni glasses, where ATLAS recovers experimentally observed short-range order (SRO) and enables steering toward target Warren-Cowley SRO parameters and optimized mechanical bulk modulus. 
Finally, we use composition-conditioned ATLAS as an amortized sampling engine for active-learning-based inverse design of high-entropy metallic glasses, aiming to jointly optimize bulk modulus, shear modulus and the Pugh ratio as measures of stiffness and ductility. 
A single pretrained ATLAS model outperforms composition-specific models trained from scratch and eliminates repeated retraining throughout the search, reducing the cumulative training cost by hundreds of fold. Coupled to a large language model (LLM)-guided evolutionary agent and the universal MACE-MPA-0 machine learning interatomic potential (MLIP), ATLAS selects up to five elements from an eight-element pool and identifies a converged Pareto frontier within 480 oracle evaluations. 
These results position ATLAS as a foundation model for amorphous materials, unifying Boltzmann sampling, thermodynamic estimation, controllable generation and composition-level inverse design within a single computational framework.

\section*{Results}

\begin{figure}[!h]
  \centering
  \includegraphics[width=\textwidth]{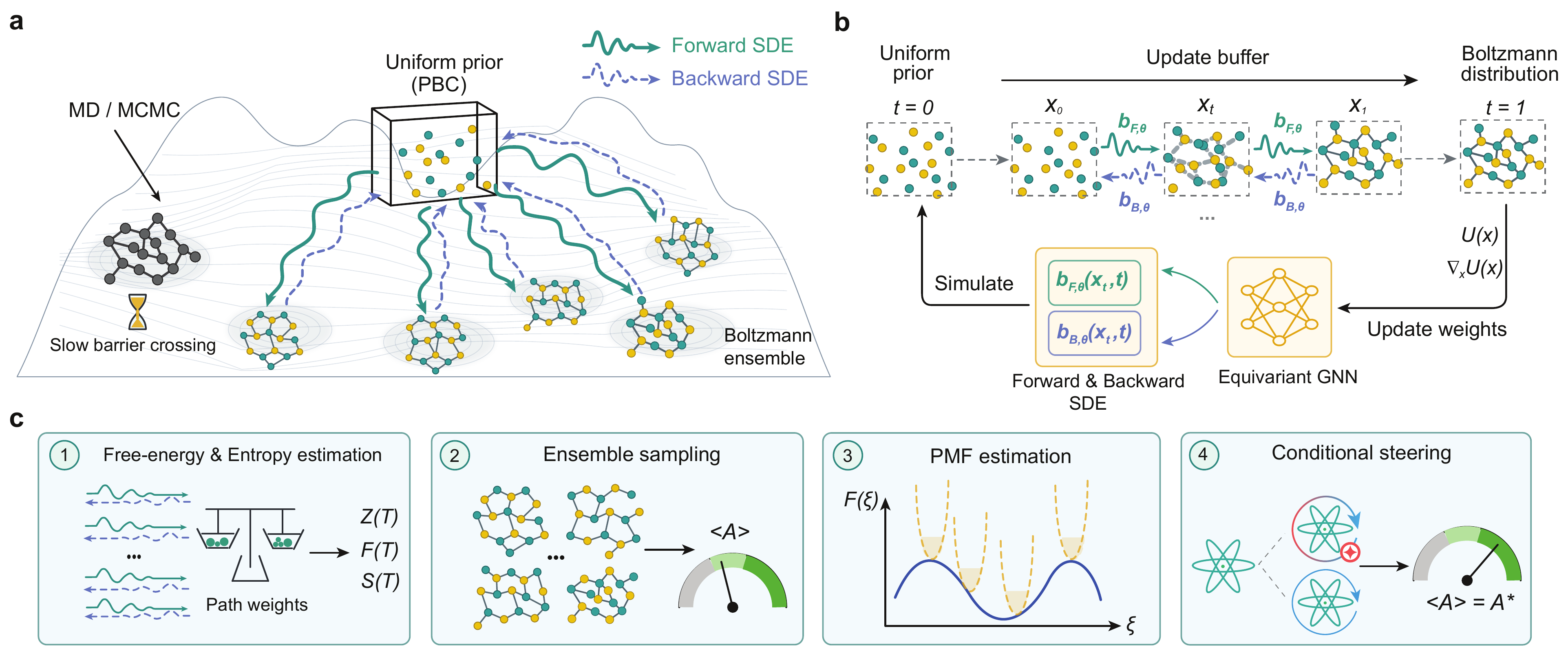}
  \caption{\textbf{Foundation neural sampling of amorphous materials.} \textbf{a.} Molecular dynamics (MD) and Markov chain Monte Carlo (MCMC) can suffer from slow barrier-crossing when traversing rugged energy landscapes. ATLAS instead learns stochastic transport, through forward and backward stochastic differential equations (SDEs), between a uniform prior under periodic boundary conditions and the Boltzmann ensemble, enabling direct generation of independent equilibrium structures without relying on precomputed data.
  \textbf{b.} Training framework of ATLAS. An equivariant graph neural network parameterizes the forward and backward diffusion processes, $b_{F,\theta}(\mathbf{x}_t,t)$ and $b_{B,\theta}(\mathbf{x}_t,t)$, along an interpolation from the prior at $t=0$ to the target distribution at $t=1$. Model weights are updated using the potential energy $U(\mathbf{x})$ and its gradient $\nabla_\mathbf{x} U(\mathbf{x})$, while newly generated configurations replenish the replay buffer.
  \textbf{c.} Applications supported by ATLAS, enabled by forward and backward path weights, including free energy and entropy estimation, equilibrium sampling and ensemble-property prediction, potential-of-mean-force (PMF) estimation along collective variables $\xi$; and conditional steering towards a target observable.}
  \label{fig1}
\end{figure}

\subsection*{Overview of ATLAS} 
Here we present an overview of ATLAS in Fig.\,\ref{fig1}.
Our aim is to draw independent configurations from the canonical (NVT) ensemble of an $N$-atom system in a $d$-dimensional supercell with periodic boundary conditions (Fig.\,\ref{fig1}a),
\begin{equation}
    p_{\mathrm{target}}(\mathbf{x}) = \frac{1}{Z}\,e^{-\beta U(\mathbf{x})}, \qquad \beta = \frac{1}{k_{\mathrm B}T},
\end{equation}
where $\mathbf{x}\in(\mathbb{R}^d/\mathbb{Z}^d)^N$ denotes atomic coordinates on the flat torus and $U(\mathbf{x})$ is a differentiable interatomic potential evaluated under full periodic boundary conditions. ATLAS requires two quantities from the target distribution: the unnormalized log density, $\log p_{\mathrm{target}}(\mathbf{x})=-\beta U(\mathbf{x})+\mathrm{const}$, and its score, $\nabla_{\mathbf{x}}\log p_{\mathrm{target}}(\mathbf{x})=-\beta\nabla_{\mathbf{x}}U(\mathbf{x})=\beta\mathbf{F}(\mathbf{x})$. For the particle systems studied here, these quantities are obtained directly from interatomic potential energies and forces.

As shown in Fig.\,\ref{fig1}b, ATLAS starts with a stochastic interpolant on the torus that connects a simple prior $p_0$ to the target distribution $p_{\mathrm{target}}$ over a time interval $t\in[0,1]$. We choose $p_0$ as the uniform distribution over the simulation cell, restricted to the zero-center-of-mass subspace. Intermediate configurations are constructed as a wrapped interpolant between a prior draw $\mathbf{x}_0\sim p_0$ and a terminal configuration $\mathbf{x}_1$,
\begin{equation}
\mathbf{x}_t=\mathrm{wrap}\!\big(\mathbf{x}_0+\alpha(t)\,\boldsymbol{\Delta}+\gamma(t)\,\mathbf{z}\big),
    \qquad
    \boldsymbol{\Delta}=\mathrm{wrap}(\mathbf{x}_1-\mathbf{x}_0),
\end{equation}
where $\mathbf{z}$ denotes Gaussian noise, $\alpha(t)$ monotonically increases from $\alpha(t=0)=0$ to $\alpha(t=1)=1$, and $\gamma(t)$ is an interpolant width that vanishes at both endpoints, i.e. $\gamma(t=0)=\gamma(t=1)=0$. The wrapping operation under the minimum-image convention ensures that all interpolated configurations remain inside the periodic cell.
This interpolant defines time-dependent marginals connecting $p_0$ and $p_{\mathrm{target}}$, which can be realized by forward and backward stochastic differential equations (SDEs),
\begin{equation}
\label{eq: fwd bwd SDE}
    \mathrm{d}\mathbf{x}_t=b_{\mathrm F}(\mathbf{x}_t,t)\,\mathrm{d}t+\sigma(t)\,\mathrm{d}\mathbf{w}_t,\qquad
    \mathrm{d}\mathbf{x}_t=b_{\mathrm B}(\mathbf{x}_t,t)\,\mathrm{d}t+\sigma(t)\,\mathrm{d}\bar{\mathbf{w}}_t ,
\end{equation}
where $\mathbf{w}_t,\bar{\mathbf{w}}_t$ are the standard Wiener processes, $\sigma(t)$ is a prescribed noise schedule, and $b_{\mathrm F}$ and $b_{\mathrm B}$ are the forward and backward drifts. ATLAS learns both drifts ($b_{\mathrm F,\theta}(\mathbf{x}_t,t),b_{\mathrm B,\theta}(\mathbf{x}_t,t)$) with an $E(3)$-equivariant GNN (see Methods), where $b_{\mathrm F}$ generates configurations from the prior to the target, while $b_{\mathrm B}$ describes the reverse dynamics that runs backward in time. The reverse dynamics is essential for ATLAS when we estimate free energy, entropy, and ensemble averages, and conditionally steer the samples from the learned model.

Because exact Boltzmann samples are unavailable, we train ATLAS by fixed-point bootstrapping, following the strategies developed in bridge-matching samplers \cite{blessing2026bridge}, as is shown in Fig.\,\ref{fig1}b. At each iteration, the current sampler generates terminal configurations $\mathbf{x}_1$, the interatomic potential evaluates $U(\mathbf{x}_1)$ and $\mathbf{F}(\mathbf{x}_1)$, and these quantities define the target score $\nabla_{\mathbf{x}_1}\log p_{\mathrm{target}}(\mathbf{x}_1)=\beta\mathbf{F}(\mathbf{x}_1)$. For each interpolant $\mathbf{x}_t$, we construct a force-informed score target $\hat{s}(\mathbf{x}_t,t)$ which depends on $\beta\mathbf{F}(\mathbf{x}_1)$ and gives drift targets $\hat{b}_{\mathrm F}(\mathbf{x}_t,t)$ and $\hat{b}_{\mathrm B}(\mathbf{x}_t,t)$   (see Methods). 
During the learning process, ATLAS updates the weights of the $E(3)$-equivariant GNN to match the corresponding force-informed targets. Repeating this iteration cycle shown in Fig.\,\ref{fig1}b drives the sampler toward a self-consistent neural sampler. 
We note that ATLAS can also be generalized over thermodynamic and chemical conditions by passing the temperature $T$ and multi-element composition vector $\mathbf{c}=\{c_j\}$ directly as conditioning inputs to $b_{\mathrm F,\theta}(\mathbf{x}_t,t,\mathbf{c},T)$ and $b_{\mathrm B,\theta}(\mathbf{x}_t,t,\mathbf{c},T)$, with $0\le c_j\le 1$ and $\sum_j c_j=1$. A single network can therefore learn a family of amorphous ensembles rather than a separate sampler for each system. 
Furthermore, because the drifts are represented by local GNNs, models trained on moderate-size supercells can be readily extrapolated to larger systems during evaluation.

Powered by this flexible architecture and training framework, ATLAS supports a broad range of downstream tasks (Fig.\,\ref{fig1}c). First, the learned forward and backward path weights enable estimation of temperature-dependent partition function $Z(T)$, free energy $F(T)$ and entropy $S(T)$. 
Second, independently generated configurations provide ensemble estimates of structural and thermodynamic observables under the target Boltzmann distribution. 
Third, inference-time steering can efficiently explore rare regions of configuration space and recover potentials of mean force (PMF) $F(\xi)$ along prescribed collective variables $\xi$. 
Fourth, a pretrained ATLAS model can be fine-tuned or steered at inference time toward target expectation values through reward tilting, enabling observable-guided generation. 
By amortizing sampling across composition, ATLAS provides an efficient engine for multi-objective inverse design in complex chemical spaces. These capabilities are demonstrated in the following sections, with full algorithmic details, training procedures and parameter settings provided in Methods.

\begin{figure}[!h]
  \centering
  \includegraphics[width=\textwidth]{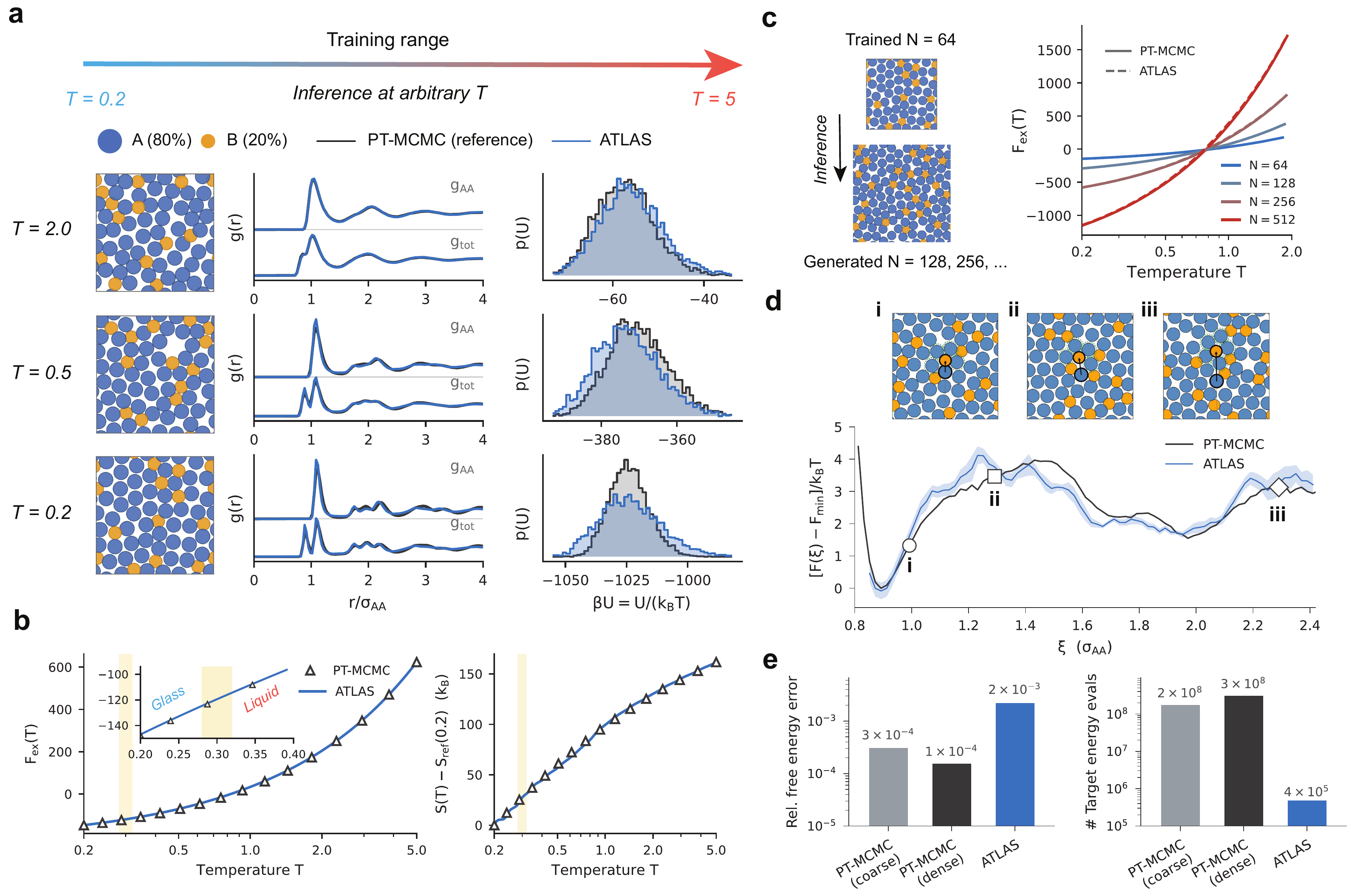}
  \caption{\textbf{Efficient sampling and estimation of ATLAS on two-dimensional Kob-Andersen systems.}
  \textbf{a.} A single ATLAS model is trained across temperatures $T=0.2$ to $5$ and evaluated at arbitrary temperatures within this range. Representative configurations, partial and total radial distribution functions, $g_{\mathrm{AA}}(r)$ and $g_{\mathrm{tot}}(r)$, and energy distributions $p(U)$ are compared with parallel tempering Markov chain Monte Carlo (PT-MCMC) at $T=2.0$, $0.5$ and $0.2$.
  \textbf{b.} Temperature-dependent excess Helmholtz free energy $F_\text{ex}(T)$ (in reduced Lennard-Jones units) and entropy relative to that at $T=0.2$, $S(T)-S_{\mathrm{ref}}(0.2)$, estimated using ATLAS and PT-MCMC. The inset enlarges the low-temperature region spanning the glass-liquid crossover, with yellow shade near $T=0.3$.
  \textbf{c.} Size transferability of ATLAS trained on $N=64$ particles and applied without retraining to systems containing $N=128$, $256$ and $512$ particles. The corresponding excess free energies per particle are compared with size-matched PT-MCMC references.
  \textbf{d.} Inference-time sampling of a cage-escape event. A tagged type-B particle (orange) moves from (i) its initial cage, through (ii) a transition region, into (iii) a neighboring metastable cage. Sequential Monte Carlo (SMC) steering along the displacement coordinate $\xi$ is compared against PT-MCMC on the potential of mean force (PMF) estimation $F(\xi)$ along $\xi$. The blue shading denotes the estimated uncertainty from ATLAS.
  \textbf{e.} Comparison of the relative free energy error and number of target-energy evaluations for coarse and dense PT-MCMC calculations and ATLAS at $T=0.2$. Errors are evaluated relative to a converged PT-MCMC reference obtained using $\simeq 4.5\times10^{8}$ target energy evaluations.}
  \label{fig2}
\end{figure}

\subsection*{Efficient sampling and estimation for Kob-Andersen models}
To validate our approach, we first benchmark ATLAS on the two-dimensional 80:20 Kob-Andersen (KA) mixture at number density $\rho=1.0$ (Methods). Instead of training an independent sampler at each thermodynamic state, we condition a single ATLAS model on temperature and train it jointly over the continuous range $T\in[0.2,5]$ (Fig.\,\ref{fig2}a) \cite{schebek2024efficient,lee2026predicting}. 
Once this one-time training cost has been amortized, the model can directly generate equilibrium configurations at arbitrary temperatures within this range, thus avoiding repeated temperature-specific simulations for re-training.

As is shown in Fig.\,\ref{fig2}a, ATLAS accurately captures the substantial structural evolution from the high-temperature liquid ($T=2$) to the low-temperature ($T=0.2$) glass. Upon cooling through the glass-transition region near $T\simeq0.3$ \cite{li2016glass}, the pair-correlation functions $g(r)$ develop sharper and increasingly structured coordination peaks, reflecting the enhancement of local order in the glassy structures. 
Across the investigated temperatures, the correlation functions obtained from ATLAS closely follow those from reference PT-MCMC results. More stringently, ATLAS reproduces the corresponding potential energy distributions $p(U)$, including their temperature-dependent average values and substantial probability overlap with PT-MCMC.
Matching these distributions is considerably more demanding than reproducing low-order structural correlations, because small errors in interatomic distances can generate large energy deviations through the steep Lennard-Jones-type repulsive potential, and geometrically plausible configurations do not necessarily occur with the correct Boltzmann probabilities \cite{grenioux2025riemannian}. The simultaneous agreement in structure and energy therefore shows that ATLAS closely approximates the target equilibrium distribution.

The path weights from forward and backward SDEs provide an absolute anchor at the highest temperature, $T=5$, with MBAR over the temperature-conditioned samples further providing estimates of the temperature-dependent partition function $Z(T)$ and, consequently, the free energy $F(T)=-k_{\mathrm B}T\log Z(T)$ and entropy $S(T)=(\langle E\rangle-F)/T$ (more details in Methods). As shown in Fig.\,\ref{fig2}b, both quantities closely agree with the PT-MCMC references across the full training interval. 
We note that here we plot the excess free energy $F_\text{ex}(T)=F(T)-F_\text{ideal}(T)$, where $F_\text{ideal}(T)$ is the Helmholtz free energy for the ideal gas (see Methods for more details). For simplicity we refer this excess free energy as the free energy in the following.
A magnified view of the low-temperature regime reveals a subtle change in the slope of $F_\text{ex}(T)$ upon warming above $T=0.2$, accompanied by a crossover in $S(T)$ near $T\simeq 0.3$. This thermodynamic behavior is consistent with previous numerical estimates of the glass-transition crossover in the two-dimensional KA system \cite{li2016glass}.
Moreover, the local GNN architecture of ATLAS also permits inference beyond the system size used during training \cite{schebek2026scalable}. As shown in Fig.\,\ref{fig2}c, an ATLAS model trained exclusively on systems with $N=64$ particles can be applied without retraining to $N=128$, $256$ and $512$ particles per supercell. The resulting free energy estimates remain close to corresponding PT-MCMC references with less than $\sim$3\% max relative error at every temperature up to $N=256$.
This size transfer enables training on small systems, for which reference sampling is comparatively accessible, followed by deployment on substantially larger systems. Beyond equilibrium thermodynamics, the same path weight formalism can be combined with inference-time steering to resolve rare activated processes and their PMFs (Fig.\,\ref{fig2}d), as discussed in the next subsection.

Finally, we quantify the trade-off between thermodynamic accuracy and computational cost in Fig.\,\ref{fig2}e. At the glass-forming temperature $T=0.2$, ATLAS achieves a relative free energy error of $2\times10^{-3}$ using only $\sim4\times10^{5}$ target-energy evaluations, compared with $2\times10^{8}$--$3\times10^{8}$ evaluations for PT-MCMC, depending on the trajectory length. ATLAS therefore attains below $0.2\%$ free energy error with at least 500-fold fewer energy evaluations, demonstrating its efficiency in generating independent and thermodynamically reliable amorphous configurations.

\subsection*{Conditional steering of ATLAS for target properties}
Beyond drawing equilibrium configurations from the target Boltzmann distribution, ATLAS can be used as a conditional generator on expectation values. We formulate this problem as reward tilting of the target ensemble. Specifically, given a pretrained ATLAS model for the base distribution $p_\text{base}(\mathbf{x})\propto \exp[-\beta U(\mathbf{x})]$, we define a tilted ensemble
\begin{equation}
    p_{\eta}(\mathbf{x})=Z_\eta^{-1}p_\text{base}(\mathbf{x})\exp[\eta r(\mathbf{x})],
\end{equation}
where $r(\mathbf{x})$ is a user-defined reward function, $\eta$ controls the strength of the tilt, and $Z_\eta$ is the normalizing constant. Equivalently, reward tilting modifies the potential energy surface as $U_\eta(\mathbf{x})=U(\mathbf{x})-\eta r(\mathbf{x})/\beta$.

We consider two complementary ways to realize this tilted ensemble, as summarized in Fig.\,\ref{fig3}a. First, the pretrained ATLAS model can be directly fine-tuned toward the tilted ensemble. In this setting, the reward is evaluated at the terminal configuration $\mathbf{x}_1$, and the model parameters are updated so that generated samples have increased probability under $p_\eta$. This strategy is natural when $r(\mathbf{x})$ can be designed as a differentiable function and when a reusable specialized generator is desired.
When retraining is costly relative to inference-time sampling, we instead use sequential Monte Carlo (SMC) for training-free tilting during generation \cite{wu2023practical,he2025rne}. SMC propagates a population of weighted particles along the generation trajectory, progressively reweighting and resampling them according to an intermediate reward function $r_t(\mathbf{x}_t)=\alpha(t)\,\eta\,r(\mathbf{x}_t)$, where $\alpha(t)$ is an annealing schedule. For differentiable rewards w.r.t. the configuration, the proposal can additionally include a gradient guidance term; for non-differentiable rewards, SMC operates through reward-based reweighting and resampling alone.

The approaches above assume that an intermediate reward $r(\mathbf{x}_t)$ can be evaluated along the generation trajectory. This is often impractical for mechanical or functional properties of amorphous materials, which are meaningful only after the generated structure has been cleaned and relaxed. The reward is therefore naturally defined only on the terminal configuration $\mathbf{x}_1$.
For such terminal-only rewards $r(\mathbf{x}_1)$, we estimate the final clean structure as $\hat{\mathbf{x}}_1(\mathbf{x}_t,t)=\mathbb{E}[\mathbf{x}_1\mid\mathbf{x}_t,t]$ and define $r_t(\mathbf{x}_t)\approx r(\hat{\mathbf{x}}_1(\mathbf{x}_t,t))$. In ATLAS, $\hat{\mathbf{x}}_1$ is obtained from the learned network through Tweedie’s formula \cite{robbins1992empirical}. This enables SMC to guide or reweight intermediate particles using rewards evaluated on their predicted terminal structures. Further details are provided in Methods.

\begin{figure}[!h]
  \centering
  \includegraphics[width=0.95\textwidth]{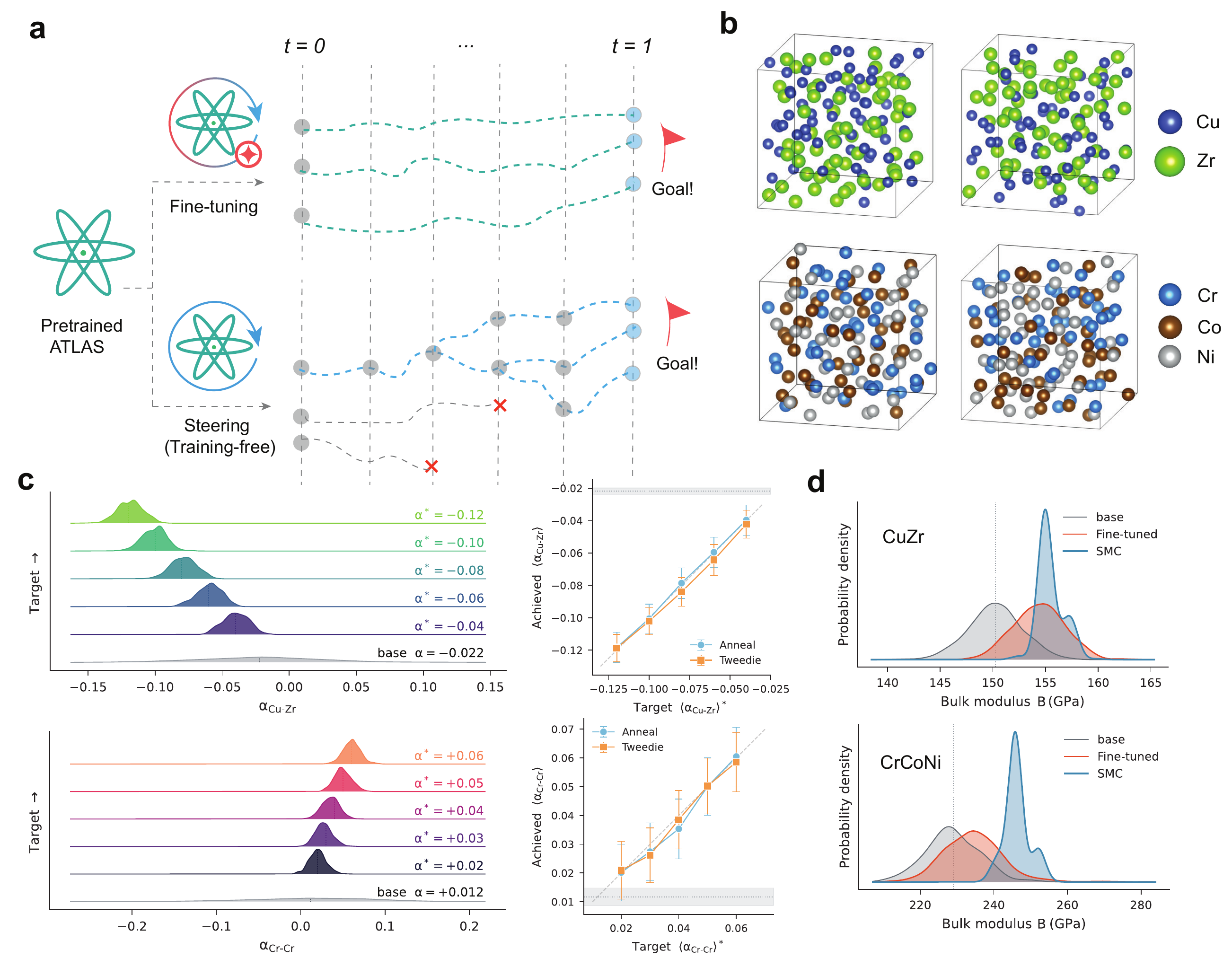}
  \caption{\textbf{Conditional generation and property tilting with ATLAS on metallic glasses.} 
  \textbf{a.} Schematic of conditional generation from a pretrained ATLAS model. The sampler can be adapted by direct fine-tuning or steered at inference time without retraining. In the training-free route, particles evolve along the generative trajectory and are progressively biased toward a target reward, with unsuccessful trajectories downweighted or removed and successful trajectories enriched.
  \textbf{b.} Representative generated configurations for Cu-Zr and Cr-Co-Ni metallic glasses.
  \textbf{c.} Short-range-order (SRO) conditioning using the Warren-Cowley parameter $\alpha$. Left panels: target values of Cu-Zr ordering in Cu-Zr (top) and Cr-Cr ordering in Cr-Co-Ni (bottom) are imposed over a sequence of $\alpha^*$ values, producing controlled shifts of the generated $\alpha$ distributions under the annealed-reward SMC steering. Right panels compare achieved and target SRO values for annealed and Tweedie-style inference-time steering. Error bars denote the standard deviation of sampled $\alpha$.
  \textbf{d.} Property-guided generation by tilting towards higher bulk modulus. Bulk-modulus distributions generated by the base ATLAS model are compared with those obtained after direct fine-tuning and SMC-based inference-time steering, for both Cu-Zr and Cr-Co-Ni metallic glasses.}
  \label{fig3}
\end{figure}

We first demonstrate inference-time steering by estimating the PMF for an activated cage-escape event in the two-dimensional KA system at $T=0.2$ (Fig.\,\ref{fig2}d). We tag a type-B particle and a nearby type-A particle and define the reaction coordinate $\xi$ as the axial component of their minimum-image separation. The reaction coordinate is partitioned into overlapping umbrella windows, each sampled by steering the pretrained ATLAS model toward the corresponding biased ensemble using SMC.
The resulting weights correct the guided proposals and provide the relative normalization of the windows, which are reweighted and combined to obtain the PMF $F(\xi)$. As shown in Fig.\,\ref{fig2}d, ATLAS closely reproduces the reference PT-MCMC result, yielding a cage-escape free energy barrier of $\Delta F^\ddagger=4.20\pm0.38\,k_{\mathrm B}T$, compared with $3.97\,k_{\mathrm B}T$ from the reference calculation.

We next apply ATLAS to two representative metallic-glass systems, binary Cu-Zr and ternary Cr-Co-Ni, using embedded-atom method potentials \cite{daw1984embedded} as the target energy models (see Methods). 
Representative independently generated configurations are shown in Fig.\,\ref{fig3}b. As a baseline, the pretrained ATLAS models capture the contrasting SRO tendencies of the two alloys, as reflected by the gray distributions in Fig.\,\ref{fig3}c. We quantify chemical ordering using the Warren-Cowley parameter,
\begin{equation}
\alpha_{ij}=1-\frac{P_{ij}}{c_j},
\label{SRO}
\end{equation}
where $P_{ij}$ is the probability that a nearest neighbor of an atom of species $i$ is of species $j$, and $c_j$ is the global concentration of species $j$. The base Cu-Zr ensemble has $\alpha_{\mathrm{Cu-Zr}}=-0.022$, indicating preferential Cu-Zr coordination, consistent with experimental observations of enhanced Zr coordination around Cu atoms in Cu-Zr metallic glasses \cite{ma2009efficient}. 
By contrast, the base Cr-Co-Ni ensemble has a weakly positive $\alpha_{\mathrm{Cr-Cr}}=+0.012$, indicating a slight suppression of nearest-neighbor Cr-Cr pairs, qualitatively matching experimental observations that Cr-Cr nearest-neighbor pairs are disfavored in Cr-Co-Ni alloys \cite{zhang2020short}.

We then condition the generated ensembles toward prescribed SRO values using the reward $r(\mathbf{x}_1)=-[\alpha_{ij}(\mathbf{x}_1)-\alpha_{ij}^{*}]^2$, which penalizes deviations from the target SRO. This formulation is particularly useful when experimental SRO measurements are available, because the measured value can be imposed directly to generate atomistic ensembles consistent with experiment for subsequent structural and property calculations \cite{zarrouk2025molecular}. 
It can also be used to diagnose or correct an imperfect interatomic potential when its unconstrained ensemble fails to reproduce the observed chemical ordering.
Because the Warren–Cowley parameter depends only on the instantaneous atomic configuration and species identities, it can be evaluated directly on the noisy intermediate state $\mathbf{x}_t$, providing a well-defined annealed reward $r_t$. 
As shown in Fig.\,\ref{fig3}c, inference-time tilting produces systematic shifts of the full SRO distributions across target ranges of $\alpha_{\mathrm{Cu-Zr}}^{*}=-0.04$ to $-0.12$ and $\alpha_{\mathrm{Cr-Cr}}^{*}=+0.02$ to $+0.06$.
Alternatively, the reward can be evaluated on the Tweedie estimate of the clean terminal structure, $r(\hat{\mathbf{x}}_1(\mathbf{x}_t,t))$. The right panels of Fig.\,\ref{fig3}c show that the achieved mean SRO values from both SMC strategies closely track the prescribed targets. 
These results demonstrate that SRO can be controlled effectively at inference time, even when the reward is evaluated on strongly perturbed intermediate configurations, without any model fine-tuning.

Beyond SRO, which can be computed directly from atomic coordinates, we then consider the more challenging task of bulk-modulus optimization.
The bulk modulus $B$ is obtained only after structural relaxation and an equation-of-state calculation involving multiple potential-energy evaluations. It is therefore expensive to evaluate repeatedly and is not directly differentiable with respect to the generated coordinates. We address this limitation by approximating a differentiable surrogate model $\hat{B}(\mathbf{x}_1)$ (see Methods) and using its prediction as the terminal reward, $r(\mathbf{x}_1)=\hat{B}(\mathbf{x}_1)$, to be maximized. With $\eta$ denoting the tilting strength, this defines the tilted ensemble
$p_{T,\eta}(\mathbf{x}_1)\propto p_T(\mathbf{x}_1)\exp[\eta\hat{B}(\mathbf{x}_1)].$
The same surrogate can be used either to fine-tune the pretrained sampler or to steer it at inference time using SMC. Both strategies systematically shift the generated ensembles toward larger bulk moduli, as shown in Fig.\,\ref{fig3}d. However, they produce qualitatively different tilted ensembles. Fine-tuning yields a smoother displacement of the original distribution and largely preserves its variance and overall shape. By contrast, SMC more strongly enriches high-modulus configurations, but can substantially concentrate its generation within a narrower subset of structures. Therefore, fine-tuning is more conservative when preserving ensemble diversity is important, whereas SMC is better suited to efficiently enriching rare, high-performing configurations.
The SRO-matching and bulk-modulus-tilting results above establish ATLAS as a flexible conditional sampler that can accommodate rewards defined on noisy intermediate states, terminal-only objectives and non-differentiable materials properties that are computationally costly to evaluate.

\subsection*{ATLAS as a foundation model to navigate compositional inverse design}
Beyond generating amorphous ensembles at a fixed composition, ATLAS can serve as an amortized sampling engine for materials inverse design across chemical composition space (Fig.\,\ref{fig4}a). 
Rather than training an independent sampler for every alloy, we pretrain a single composition-conditioned ATLAS model over a broad set of chemical elements and reuse it across chemically diverse systems at the same target temperature, with the underlying energetics provided by either a classical interatomic potential or a higher-fidelity pretrained MLIP. For each proposed composition $\mathbf{c}$, ATLAS generates an ensemble of amorphous structures ${\mathbf{x}_i}$, from which ensemble-averaged target properties are evaluated. 
These target properties are returned to an agentic LLM evolutionary algorithm (LLM-EA) framework, which reasons over the accumulated search history, proposes new compositions, and asynchronously launches and monitors multiple sampling and property-evaluation jobs (Methods). 
The LLM agent naturally supports multi-objective design by evaluating candidates from their full property vectors rather than a fixed scalarization, retaining non-dominated solutions, and proposing new compositions that expand the Pareto frontier for the multiple targets.

\begin{figure}[!h]
  \centering
  \includegraphics[width=\textwidth]{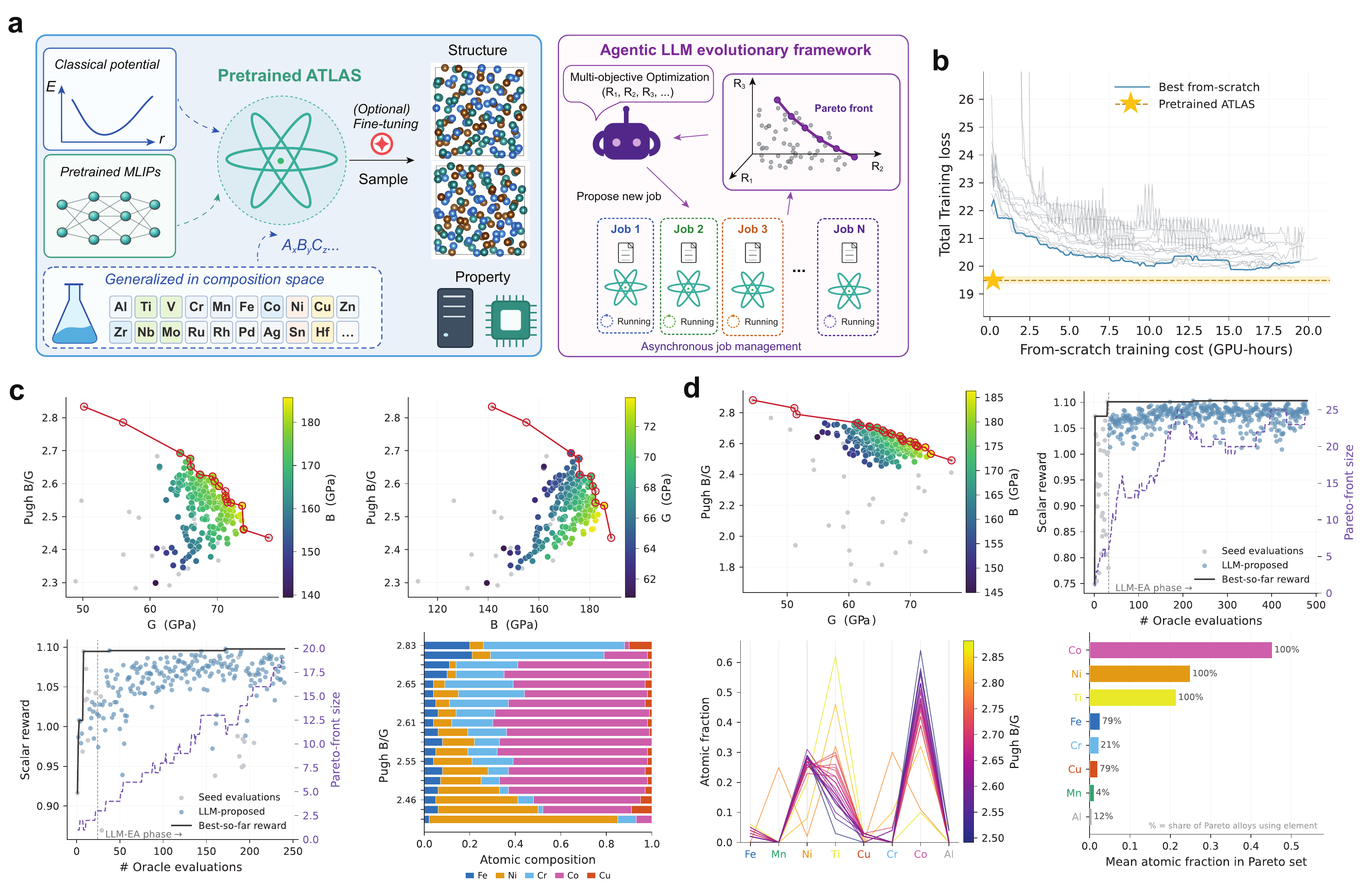}
  \caption{{\textbf{Agentic inverse design of amorphous materials with ATLAS as a foundation model.}
  \textbf{a.} Schematic of ATLAS-driven composition-space optimization. A pretrained ATLAS model is amortized over chemical compositions and coupled to either classical force fields or MLIPs. For each proposed composition, ATLAS samples amorphous structures and returns ensemble-evaluated target properties to an LLM-based agentic optimizer, which is prompted to propose asynchronous sampling jobs and iteratively searches for Pareto-optimal compositions.
  \textbf{b.} Comparison between from-scratch training and pretrained ATLAS initialization for compositions drawn from the Fe-Mn-Ni-Ti-Cu-Cr-Co-Al pool. Blue: best from-scratch training trajectory from 16 randomly scanned compositions. Yellow dashed line and shading: zero-shot loss of foundation ATLAS and its standard deviation.
  \textbf{c.} LLM-based evolutionary optimization in the five-component Fe-Ni-Cr-Co-Cu composition space using an EAM force field, jointly optimizing bulk modulus $B$, shear modulus $G$, and Pugh ratio $B/G$. Top: discovered trade-offs in the $G$-$B/G$ and $B$-$B/G$ subspaces, with red markers indicating the Pareto frontier. Bottom: scalar reward (see main text) and Pareto-front size versus oracle evaluations, and representative Pareto-optimal glass compositions.
  \textbf{d.} Extension of \textbf{c} to selecting up to five elements from the eight-element Fe-Mn-Ni-Ti-Cu-Cr-Co-Al pool and optimizing their compositions using the MACE-MPA-0 potential. Top: Pareto frontier in the $G$--$B/G$ space and improvement of scalar reward and Pareto-front size during the search. Bottom: atomic fractions of Pareto-optimal glasses colored by $B/G$, and mean atomic fraction of each element in the Pareto set, with percentages indicating the fraction of Pareto glasses containing that element.}}
  \label{fig4}
\end{figure}

We first couple ATLAS to the LLM-EA framework in the Fe-Ni-Cr-Co-Cu space using an EAM potential (Fig.\,\ref{fig4}c). The design objective is to identify high-entropy metallic glasses that jointly optimize the bulk modulus $B$, shear modulus $G$, and Pugh ratio $B/G$, thereby balancing resistance to volumetric compression, shear rigidity and hardness, and ductility.
The agentic framework maintains up to $N=32$ asynchronous ATLAS sampling and property-evaluation jobs. As jobs finish, it collects the resulting $B$, $G$, and $B/G$ values and returns them to the LLM. At each update, the accumulated candidates are downsampled with priority determined by their distance to the current Pareto frontier, such that near-front solutions are retained more frequently while more distant candidates are sampled less often. 
Conditioned on this filtered search history, the LLM is prompted to improve all three objectives jointly while preserving exploration of under-sampled regions of composition space, and then proposes a new composition to replace the completed job. The full details of LLM-EA implementation are provided in Methods.
For visualization of the optimization trajectory, we additionally report the scalar reward which is not used as the optimization objective, $R_{\mathrm{vis}}=0.25\times B/B_0+0.25\times G/G_0+0.5\times (B/G)/(B_0/G_0)$, where $B_0$, $G_0$ and $(B/G)_0$ are reference values used to normalize the three objectives.
As shown in Fig.\,\ref{fig4}c, LLM-EA resolves well-defined Pareto frontiers in both the $G$-$B/G$ and $B$-$B/G$ projections, revealing a clear trade-off between stiffness and the ductility proxy. The Pareto set also exhibits a systematic chemical progression: Cr-rich compositions populate the high-$B/G$ region, whereas Ni- and Co-rich alloys dominate the high-$G$ and high-$B$ end of the frontier, with Fe and Cu acting mainly as lower-concentration tuning elements. 
The simultaneous growth of the best scalar score and Pareto-set size indicates that the search improves performance while expanding the range of accessible compromises. Rather than identifying a single optimum, the optimization therefore extracts chemically interpretable design rules linking composition to the stiffness-ductility trade-off.

Before scaling to even larger chemical space, we quantify the computational advantage of composition amortization (Fig.\,\ref{fig4}b). The composition-conditioned ATLAS model incurs a one-time training cost of 36 GPU hours, yet its zero-shot loss is already lower than that of the best composition-specific model trained from scratch. 
Because each from-scratch model requires approximately 20 GPU hours, a full-loop active learning campaign, such as the 480 evaluations to be shown in Fig.\,\ref{fig4}d, would otherwise cost about 9,600 GPU hours. Compared with 36 GPU hours for the amortized model, there is a nearly 300-fold reduction in terms of GPU training cost for pretrained ATLAS model. Pretraining therefore enables repeated composition-space queries without retraining.
Enabled by this amortization, we next extend the workflow to the substantially larger task of selecting up to five elements from an eight-element pool (Fe-Mn-Ni-Ti-Cu-Cr-Co-Al) and optimizing their continuous atomic fractions (Fig.\,\ref{fig4}d). Here, the target distribution is defined by the universal MACE-MPA-0 potential \cite{batatia2022mace,batatia2025foundation}, which expands the accessible chemical space without requiring a system-specific EAM potential for every elemental combination. 
Although MACE-MPA-0 is substantially more expensive to evaluate than EAM potentials, ATLAS requires only $\sim10^6$ target energy evaluations during one-time training and can subsequently be reused across the entire eight-element composition space, making the MLIP-powered inverse-design campaign computationally practical.

Despite the combinatorial element-selection problem and high-dimensional composition space, the LLM-EA search progressively improves the best scalar score while expanding the Pareto set. As shown in Fig.\,\ref{fig4}d, the candidate compositions converge to a well-defined frontier in the $G$-$B/G$ plane within 480 ATLAS evaluations and exhibit strong elemental preferences. 
Specifically, Co, Ni and Ti appear in every Pareto-optimal alloy and constitute the dominant fractions. This preference is physically plausible, because Co and Ni are relatively stiff constituent elements, whereas the substantially larger atomic radius of Ti introduces pronounced size and chemical mismatch with Co and Ni, which is commonly associated with stabilized supercooled liquids and enhanced glass-forming ability \cite{inoue2000stabilization,takeuchi2005classification}.
Additionally, Fe and Cu occur in approximately $79\%$ of the Pareto set, typically at relatively low concentrations, whereas Cr is retained in only about $21\%$ of the optimal compositions, and Mn and Al are nearly absent. 
The resulting Pareto set therefore provides not a single optimum but a batch of experimentally actionable candidates spanning different stiffness-ductility compromises, from which compositions can be selected according to the intended application and synthesized in parallel. The same ATLAS\,+\,LLM-EA workflow can also be readily transferred to alternative elemental pools and target properties. 
However, because MACE-MPA-0 and other universal MLIPs are not guaranteed to quantitatively predict accurate amorphous mechanical properties, these candidates should be regarded as model-guided hypotheses, which require validation through higher-fidelity calculations or experimental closed-loop refinement \cite{fragapane2026amorphous}.

\section*{Discussion}
In this work, we have introduced ATLAS, a foundation diffusion sampler for amorphous matter. ATLAS directly generates independent configurations from target thermodynamic ensembles, enabling ensemble-property prediction, entropy and free energy estimation, observation-constrained structure generation and property-directed sampling within the same computational framework. 
We demonstrate these capabilities across systems ranging from the Kob-Andersen glass former to binary metallic glasses and quinary high-entropy metallic glasses, using both empirical interatomic potentials and more expressive MLIPs. 
The same sampler generalizes across temperature and chemical composition, transfers across system sizes, and can be coupled to external agents for multi-objective materials design. This amortization is particularly valuable for expensive MLIPs with millions of parameters, because ATLAS can directly produce independent configurations after one-time training, instead of repeatedly simulating trajectories through femtosecond-scale MD integration.
Costly potential evaluations through MLIPs can therefore be reused across thermodynamic conditions, compositions and downstream tasks. The resulting combination of direct ensemble sampling, controllable generation and agentic optimization provides a promising route towards scalable inverse design in complex disordered materials.

Despite these advances, several limitations exist and define the next stage of development. 
Firstly, the present implementation focuses on fixed-$NVT$ ensembles, with metallic-glass volumes assigned using a composition-based mixing rule (see Methods). Extending ATLAS to the $NPT$ ensemble would allow density and local packing to respond consistently to temperature, composition and pressure. More general grand-canonical or semi-grand-canonical formulations, incorporating particle insertion, deletion and species exchange, could further enable flexible joint sampling of composition and structure within a unified generative framework.
Secondly, the applicability of ATLAS to substantially larger systems remains to be demonstrated. The supercells considered here typically contain 128 atoms, sufficient to capture local amorphous order and bulk thermodynamic properties, but too small to resolve phenomena involving extended structural features such as interfaces, dislocations and grain boundaries, which may require systems containing thousands to tens of thousands of atoms.
Finally, ATLAS samples thermodynamic configurations, bypassing the physical time evolution. Combining direct thermodynamic sampling with learned dynamical models or trajectory-based rare-event methods could provide access to relaxation pathways and dynamical heterogeneity, thereby shedding light on the physics of the glass transition and connecting the generated ensembles to mode-coupling theory and broader non-equilibrium glass phenomenology \cite{bengtzelius1984dynamics,janssen2018mode}.

Nevertheless, ATLAS provides the basis for a foundation model of disordered materials. Beyond metallic glasses, the same framework could be extended to oxide glasses, amorphous ceramics, disordered semiconductors, polymers and other multicomponent systems whenever their energetics can be described by a sufficiently reliable interatomic potential. 
Moreover, the target constraints need not remain limited to scalar structural or mechanical observables. Diffraction patterns, vibrational, electronic and optical spectra, and other experimental signals could be incorporated as vectorized constraints, allowing the sampler to refine atomistic ensembles against multiple complementary measurements. 
Likewise, ATLAS could enable multi-property inverse design of functional amorphous materials, targeting optical response, electrical transport, thermal conductivity and catalytic activity while navigating trade-offs across the corresponding Pareto frontiers.
The same design principle of ATLAS could also be applied to sample chemically disordered alloys, defect species, interfaces and beyond. 
By replacing repeated trajectory generation with an amortized and controllable sampler, ATLAS establishes a framework in which thermodynamic prediction, inverse structural inference and materials design become different queries to a single generative model of amorphous materials. Just as diffusion-based generative models have transformed \textit{de novo} protein and crystal structure design \cite{Watson2023DeND,zeni2025generative}, we anticipate that models such as ATLAS will similarly reshape materials discovery across the broader space of disordered systems.

\section*{Methods} 
\subsection*{Stochastic interpolants on the torus}
We consider $N$ atoms in a fixed periodic supercell at temperature $T$, i.e., the canonical (NVT) ensemble. In fractional coordinates, a configuration is a point $\mathbf{x}\in\mathbb{T}^{D}$ on the flat torus, where $\mathbb{T}=\mathbb{R}/\mathbb{Z}$, $D=dN$, and $d$ is the spatial dimension. The target is the Boltzmann distribution
\begin{equation}
    p_{\mathrm{target}}(\mathbf{x})=\frac{1}{Z}\,e^{-\beta U(\mathbf{x})},\qquad \beta=\frac{1}{k_{\mathrm B}T},\qquad Z=\int_{\mathbb{T}^D} e^{-\beta U(\mathbf{x})}\,\mathrm{d}\mathbf{x},
    \label{eq:m-target}
\end{equation}
where $U$ is a differentiable interatomic potential evaluated under full periodic boundary conditions and $Z$ is the partition function. ATLAS requires only point-wise evaluations of the potential $U(\mathbf{x})$ and the corresponding forces $\mathbf{F}(\mathbf{x})=-\nabla_{\mathbf{x}} U(\mathbf{x})$, which determine the target score $\nabla_{\mathbf{x}}\log p_{\mathrm{target}}(\mathbf{x})=\beta \mathbf{F}(\mathbf{x})$; no reference configurations are used at any stage.
ATLAS builds on stochastic interpolants \cite{albergo2025stochastic} which prescribe the process $(\mathbf{x}_t)_{t\in[0,1]}$ that interpolates between a prior draw $\mathbf{x}_0\sim p_0$ and a target draw $\mathbf{x}_1\sim p_{\mathrm{target}}$,
\begin{equation}
    \mathbf{x}_t=\mathrm{wrap}\!\big(\mathbf{x}_0+\alpha(t)\,\boldsymbol{\Delta}+\gamma(t)\,\mathbf{z}\big),
    \qquad
    \boldsymbol{\Delta}=\mathrm{wrap}(\mathbf{x}_1-\mathbf{x}_0),
    \quad \mathbf{z}\sim\mathcal{N}(\mathbf{0},I),
    \label{eq:m-interpolant}
\end{equation}
where $\alpha(t)$ increases monotonically from $\alpha(0)=0$ to $\alpha(1)=1$ and the interpolant width $\gamma(t)$ is non-negative for $0<t<1$ and vanishes at both endpoints, pinning $\mathbf{x}_t$ to $\mathbf{x}_0$ at $t=0$ and to $\mathbf{x}_1$ at $t=1$. Here $\mathrm{wrap}(y)=y-\lfloor y+\tfrac12\rfloor\in[-\tfrac12,\tfrac12)$, applied componentwise, folds any displacement back into the cell, so that $\boldsymbol{\Delta}$ is the minimum-image displacement between the two endpoints. Since $U$ is invariant under rigid translations of all atoms, we further remove the global translation mode by restricting all distributions to the zero-center-of-mass subspace. As prior we take the uniform distribution over the cell, $p_0=\mathrm{Unif}(\mathbb{T}^D)$.
The marginals $p_t$ defined by \eqref{eq:m-interpolant} can be realized by a whole family of stochastic processes. All of them are built from two vector fields: a velocity $b$ and a score $s=\nabla_{\mathbf{x}}\log p_t$, both conditional expectations over the interpolant,
\begin{equation}
    b(\mathbf{x},t)=\mathbb{E}\big[\dot\alpha(t)\,\boldsymbol{\Delta}+\dot\gamma(t)\,\mathbf{z}\,\big|\,\mathbf{x}_t=\mathbf{x}\big],
    \qquad
    s(\mathbf{x},t)=-\frac{1}{\gamma(t)}\,\mathbb{E}\big[\mathbf{z}\,\big|\,\mathbf{x}_t=\mathbf{x}\big],
    \label{eq:m-bs}
\end{equation}
where the dot denotes a time derivative and the second expression is the denoising score identity. The velocity transports probability mass from the prior to the target; the score corrects for injected noise. Note that the conditional expectations in \eqref{eq:m-bs} average over interpolants whose endpoints include target samples $\mathbf{x}_1\sim p_{\mathrm{target}}$ which we do not assume access to. How ATLAS trains without access to target samples is explained below. For an arbitrary noise schedule $\sigma(t)\ge 0$, the forward and backward drift $b_{\mathrm F},b_{\mathrm B}$ in \eqref{eq: fwd bwd SDE} can be written in terms of the velocity and score via
\begin{align}
    \mathrm{d}\mathbf{x}_t&=b_{\mathrm F}(\mathbf{x}_t,t)\,\mathrm{d}t+\sigma(t)\,\mathrm{d}\mathbf{w}_t=\Big[b(\mathbf{x}_t,t)+\tfrac{\sigma(t)^2}{2}\,s(\mathbf{x}_t,t)\Big]\,\mathrm{d}t+\sigma(t)\,\mathrm{d}\mathbf{w}_t,
    &&\mathbf{x}_0\sim p_0,
    \label{eq:m-fwd}\\
    \mathrm{d}\mathbf{x}_t&=b_{\mathrm B}(\mathbf{x}_t,t)\,\mathrm{d}t+\sigma(t)\,\mathrm{d}\bar{\mathbf{w}}_t=\Big[b(\mathbf{x}_t,t)-\tfrac{\sigma(t)^2}{2}\,s(\mathbf{x}_t,t)\Big]\,\mathrm{d}t+\sigma(t)\,\mathrm{d}\bar{\mathbf{w}}_t,
    &&\mathbf{x}_1\sim p_{\mathrm{target}},
    \label{eq:m-bwd}
\end{align}
which share the same marginals $p_t$ for every $\sigma$; here $\mathbf{w}_t$ and $\bar{\mathbf{w}}_t$ denote standard Wiener processes running forward and backward in time, and the state is wrapped after every integration step. Setting $\sigma=0$ recovers the deterministic probability-flow ODE. Generation amounts to integrating the forward SDE \eqref{eq:m-fwd} from uniform noise; the backward SDE \eqref{eq:m-bwd} is its time reversal and is what enables efficient free energy and entropy estimation. ATLAS learns both vector fields, $b_\theta$ and $s_\theta$, as two heads of a single equivariant graph neural network with parameters $\theta$.

\subsection*{Forces as training signal}
Neither $b$ nor $s$ is observed directly, but both are conditional expectations which can be learned by a least-squares regression: for any random vector $\hat{\mathbf{s}}$ satisfying $s(\mathbf{x},t)=\mathbb{E}[\hat{\mathbf{s}}\,|\,\mathbf{x}_t=\mathbf{x}]$---such as the denoising residual $-\mathbf{z}/\gamma(t)$ in \eqref{eq:m-bs}---minimizing $\mathbb{E}\|s_\theta(\mathbf{x}_t,t)-\hat{\mathbf{s}}\|^2$ over functions yields $s_\theta=s$.
The central identity underlying ATLAS is an alternative representation of this kind in which the score is estimated by the rescaled interatomic force: on the torus with a uniform prior, the target score identity\cite{de2024target}
\begin{equation}
    s(\mathbf{x},t)=\mathbb{E}\big[\beta\mathbf{F}(\mathbf{x}_1)\,\big|\,\mathbf{x}_t=\mathbf{x}\big]
    \label{eq:m-tsi}
\end{equation}
holds exactly for all $t$. The form of \eqref{eq:m-tsi} is specific to the torus with a uniform prior: in contrast to its Euclidean counterparts, the force enters with weight one and no time-dependent prefactor. See the Supplementary Information (SI) for a proof. In words: wherever an interpolant path visits, the interatomic force evaluated at its endpoint is an unbiased estimate of the score there. Hence, physics enters the model through forces rather than reference data. However, as in \eqref{eq:m-bs}, the expectation in \eqref{eq:m-tsi} is taken over interpolants with endpoints $\mathbf{x}_1\sim p_{\mathrm{target}}$, which are unavailable; the fixed-point iteration described next replaces them with the model's own samples.

\subsection*{Fixed-point training without reference data}
ATLAS trains by a fixed-point iteration in the spirit of bridge-matching samplers \cite{blessing2026bridge} and related fixed-point diffusion samplers \cite{havens2025adjoint,havens2026flow,liu2026adjoint}: the unavailable target samples are replaced by samples from the current model, $\mathbf{x}_1\sim p_{\mathrm{model}}$, the distribution obtained by integrating the forward SDE \eqref{eq:m-fwd} with the current parameters. Each round consists of four steps, illustrated in Fig.\,\ref{fig:m-fixedpoint}: (i) \emph{generate}: integrate \eqref{eq:m-fwd} to produce a batch of terminal configurations $\mathbf{x}_1\sim p_{\mathrm{model}}$; (ii) \emph{label}: evaluate the forces $\mathbf{F}(\mathbf{x}_1)$ of these configurations and store the pairs in a replay buffer; (iii) \emph{interpolate}: draw fresh $(\mathbf{x}_0,\mathbf{z},t)$ and form intermediate states $\mathbf{x}_t$ via \eqref{eq:m-interpolant} toward buffered configurations; (iv) \emph{regress}: update both heads by minimizing
\begin{equation}
    \mathcal{L}(\theta)=\mathbb{E}\Big[\big\|s_\theta(\mathbf{x}_t,t)-\hat{\mathbf{s}}_t\big\|^2+\big\|b_\theta(\mathbf{x}_t,t)-\hat{\mathbf{b}}_t\big\|^2\Big],
    \qquad
    \hat{\mathbf{s}}_t=\beta\mathbf{F}(\mathbf{x}_1),
    \quad
    \hat{\mathbf{b}}_t=\dot\alpha(t)\,\boldsymbol{\Delta}-\dot\gamma(t)\gamma(t)\,\hat{\mathbf{s}}_t,
    \label{eq:m-loss}
\end{equation}
with the expectation taken over $t$, $\mathbf{x}_0\sim p_0$, $\mathbf{z}\sim\mathcal{N}(\mathbf{0},I)$ and buffered $\mathbf{x}_1$. The Boltzmann density is the unique fixed point of this iteration \cite{blessing2026bridge}: if $p_{\mathrm{model}}=p_{\mathrm{target}}$, the regression targets coincide with the true fields $(b,s)$ of the interpolant and training is stationary, and no other distribution is self-consistent in this sense; see SI~\ref{si: uniqueness}. The force term in \eqref{eq:m-loss} provides the mechanism that drives the iteration toward this fixed point: each generated configuration, however far from equilibrium, is labeled with the true local direction of increasing Boltzmann probability, steering the sampled ensemble toward the target basins. Because every force evaluation is stored and reused across many gradient steps, the number of force evaluations, the dominant cost under machine-learning interatomic potentials, remains small compared to classical sampling methods, as shown in Figure \ref{fig2}.

\subsection*{Conditional training and amortization}
The fixed-point scheme applies verbatim to any family of Boltzmann targets $p_{\mathrm{target}}(\mathbf{x}\,|\,\mathbf{c})\propto e^{-\beta U(\mathbf{x};\,\mathbf{c})}$ indexed by a conditioning variable $\mathbf{c}$, which may enter through the temperature, the potential, or both. We therefore condition both heads, $b_\theta(\mathbf{x},t,\mathbf{c})$ and $s_\theta(\mathbf{x},t,\mathbf{c})$, and let each training batch draw conditions at random, generate configurations under those conditions, and label them with the corresponding forces. Hence, a single network amortizes an entire family of Boltzmann ensembles, which is in contrast to conventional simulation, where every new condition requires a separate equilibration run. In this work we explore two instances: the temperature, $\mathbf{c}=T$, which enters the training signal only through the rescaled force $\beta\mathbf{F}$, and the chemical composition, $\mathbf{c}=(c_1,\dots,c_n)$ with $c_j\ge0$ and $\sum_j c_j=1$, which enters through the species-dependent potential; our models are amortized over one of the two at a time. The trained network can then be queried at conditions never used during training. A detailed description of the training algorithm can be found in SI~\ref{si:algorithm}.

\subsection*{Free energy, entropy and ensemble averages} 
Expectation values under the target Boltzmann density are obtained from the generative model by importance reweighting, which requires the ratio $p_{\mathrm{target}}(\mathbf{x})/p_{\theta}(\mathbf{x})$ between the target and the model marginal at $t=1$. The likelihood $p_{\theta}$ is accessible through the probability-flow ODE, but each evaluation entails integrating the continuity equation along with the divergence of the learned velocity field, at a cost that precludes its use for large sample sizes. We therefore reweight on the entire diffusion path rather than on its endpoint \cite{richter2024improved,vargas2024transport,du2026feat}: Discretizing $[0,1]$ into steps $t_0=0<\dots<t_K=1$, the Euler--Maruyama transition kernels of the learned forward SDE \eqref{eq:m-fwd},
\begin{equation}
\overrightarrow{p}_k(\mathbf{x}_{k+1}\,|\,\mathbf{x}_k)=\mathcal{N}_{\mathbb{T}}\Big(\mathbf{x}_{k+1};\;\mathbf{x}_k+\big[b_\theta+\tfrac{\sigma^2}{2}s_\theta\big](\mathbf{x}_k,t_k)\,\Delta t_k,\ \sigma(t_k)^2\Delta t_k\,I\Big),
    \label{eq:m-kernels}
\end{equation}
and, analogously, the backward kernel $\overleftarrow{p}_k(\mathbf{x}_k\,|\,\mathbf{x}_{k+1})$ with drift $b_\theta-\tfrac{\sigma^2}{2}s_\theta$ evaluated at $(\mathbf{x}_{k+1},t_{k+1})$, are wrapped Gaussians on the torus with $\mathcal{N}_{\mathbb{T}}\big(\mathbf{x};\,\cdot,\,\cdot\big)
\coloneqq \sum_{\mathbf{z}\in\mathbb{Z}^d}\mathcal{N}\big(\mathbf{x}+\mathbf{z};\,\cdot,\cdot\big)$; see SI~\ref{app:torus-kernels} for further details. Each generated trajectory then carries a \textit{path weight}
\begin{equation}
    \log w(\mathbf{x}_{0:K})
    =-\beta U(\mathbf{x}_K)-\log p_0(\mathbf{x}_0)
    +\sum_{k=0}^{K-1}\log\frac{\overleftarrow{p}_k(\mathbf{x}_k\,|\,\mathbf{x}_{k+1})}{\overrightarrow{p}_k(\mathbf{x}_{k+1}\,|\,\mathbf{x}_k)},
    \label{eq:m-logw}
\end{equation}
where $\log p_0$ is a known constant for the uniform prior. 
The weights serve three purposes. First, $Z=\mathbb{E}[w]$ is an unbiased importance-sampling estimate of the partition function, from which we obtain the free energy $F(T)=-k_{\mathrm B}T\log Z$ (known as Jarzynski equalities) \cite{jarzynski1997nonequilibrium,vaikuntanathan2008escorted,du2026feat}. Second, for a batch of $M$ generated trajectories $\{\mathbf{x}^m_{0:K}\}_{m=1}^{M}$ with weights $w_m=w(\mathbf{x}^m_{0:K})$, self-normalized reweighting with $\bar w_m=w_m/\sum_{m'}w_{m'}$ 
turns any terminal observable into the reweighted ensemble average $\sum_m \bar w_m\,O(\mathbf{x}^m_K)$ that remains asymptotically unbiased under model imperfection \cite{noe2019boltzmann}. Third, combined with the entropy relation $S=(\langle U\rangle-F)/T$ and temperature amortization, a single model yields free energy and entropy curves $F(T)$, $S(T)$.
The direct path weight estimate of $Z(T)$ is most accurate at high temperatures, where the model closely matches the easily sampled liquid, but its variance grows toward low temperatures, where the sampler is imperfect, and the path weight becomes heavy-tailed. We therefore anchor the absolute free energy once at a high reference temperature, where the path weight estimate is reliable, and propagate it to colder temperatures with the multistate Bennett acceptance ratio (MBAR) \cite{shirts2008statistically}, where the temperature-conditioned model supplies approximate equilibrium samples on a set of intermediate temperatures. MBAR then combines these samples into free energy differences down to the coldest target, bypassing a direct low-temperature path weight estimate whose effective sample size collapses in the glass (see SI~\ref{si:freeenergy}). Beyond evaluation, the path weights \eqref{eq:m-logw} can also be used for the training to correct for the relative masses of basins separated by large energy barriers, to which the force-based loss is blind; we detail this importance-weighted training scheme in SI~\ref{SI: IW training}.

In addition to the configurational contribution described above, the full classical partition function contains an analytical kinetic contribution. For example, in a $d$-dimensional binary system with $N_A$ and $N_B$ particles,
\begin{equation}
Z(T,V)=\frac{1}{N_A!N_B!\Lambda_A^{dN_A}\Lambda_B^{dN_B}}Z_{\mathrm{conf}}(T,V),
\end{equation}
where $Z_{\mathrm{conf}}(T,V)=\int_V d\mathbf{r}^{N}\,e^{-\beta U(\mathbf{r}^{N})},\, N=N_A+N_B$ is tractable as described above, and $\Lambda_\alpha=\frac{h}{\sqrt{2\pi m_\alpha k_{\mathrm B}T}}$ is the thermal de Broglie wavelength of species $\alpha$. 
Although the kinetic contribution is analytical, its absolute value depends on the Planck constant $h$, which is not defined in Lennard-Jones reduced units and therefore introduces an arbitrary additive reference to the absolute free energy.
We therefore report the excess free energy relative to an ideal gas with the same composition, volume, and temperature. The corresponding binary ideal-gas partition function is
\begin{equation}
Z_{\mathrm{ideal}}(T,V)=\frac{V^N}{N_A!N_B!\Lambda_A^{dN_A}\Lambda_B^{dN_B}}.
\end{equation}
The excess Helmholtz free energy is then
\begin{equation}
F_{\mathrm{ex}}(T)=F(T)-F_{\mathrm{ideal}}(T)=-k_{\mathrm B}T\ln\left[\frac{Z_{\mathrm{conf}}(T,V)}{V^N}\right].
\end{equation}
Thus, if the configurational free energy obtained above is $F_{\mathrm{conf}}(T)=-k_{\mathrm B}T\ln Z_{\mathrm{conf}}(T,V)$, then the excess free energy is simply $F_{\mathrm{ex}}(T)=F_{\mathrm{conf}}(T)+Nk_{\mathrm B}T\ln V.$

\subsection*{Fine-tuning and inference-time steering to tilted ensembles}
It is often desirable to sample not the Boltzmann ensemble itself but a tilted version of it, $p_\eta(\mathbf{x})\propto p_{\mathrm{target}}(\mathbf{x})\,e^{\eta r(\mathbf{x})}$, for instance, to concentrate the ensemble on configurations with a prescribed short-range order or a desired property. Equivalently, $p_\eta$ is a Boltzmann distribution on the biased landscape $U_\eta=U-\eta r/\beta$. We consider two routes for generating approximate samples from $p_\eta$: fine-tuning the pretrained model toward the tilted ensemble, and steering the unmodified model at inference time.

For fine-tuning, note that when the reward is differentiable, the score of the tilted ensemble, $\beta\mathbf{F}+\eta\nabla_{\mathbf{x}} r$, is available point-wise, so the fixed-point scheme applies verbatim with $p_\eta$ in place of $p_{\mathrm{target}}$: we simply rerun training, initialized from the pretrained checkpoint. This route is natural when a reusable specialized generator is desired.

As an alternative, ATLAS can be steered at inference time via Radon--Nikodym estimator (RNE) based on the path weight \cite{he2025rne}: A population of weighted particles is initialized from the prior and propagated by a proposal SDE with drift $a(\mathbf{x},t)$ and transition density $\overrightarrow{q}_k(\mathbf{x}_{k+1}\,|\,\mathbf{x}_k)=\mathcal{N}_{\mathbb{T}}\big(\mathbf{x}_{k+1};\;\mathbf{x}_k+a(\mathbf{x}_k,t_k)\,\Delta t_k,\ \sigma^2\Delta t_k\,I\big)$, accompanied by an intermediate reward $r_t$ that satisfies the boundary conditions $r_0\equiv0$ and $r_1=\eta r$; the first condition makes the initialization from the prior exact, the second makes the final target the desired tilt. Within these constraints, both $a$ and $r_t$ can be chosen freely without biasing the method---the choices affect only the variance of the weights, and hence the efficiency. After each step $\mathbf{x}_k\to\mathbf{x}_{k+1}$ of the proposal, the weight of a particle is updated as 
\begin{align}
    \log w(\mathbf{x}_{0:k+1})
    =\log w(\mathbf{x}_{0:k})
    +{r_{t_{k+1}}(\mathbf{x}_{k+1})-r_{t_k}(\mathbf{x}_k)}
    +{\log \frac{\overrightarrow{p}_k(\mathbf{x}_{k+1}\,|\,\mathbf{x}_k)}{\overrightarrow{q}_k(\mathbf{x}_{k+1}\,|\,\mathbf{x}_k)}},
    \label{eq:m-smc}
\end{align}
starting from $w\equiv1$. Equation \eqref{eq:m-smc} follows from the same construction as the path weight \eqref{eq:m-logw}: tilting the terminal target and replacing the learned forward kernels by the proposal leaves the backward kernels of the construction free, and choosing them as the time reversal of the learned chain cancels all intractable terms, so that only the reward increments and the log-ratio of forward transition densities survive (see SI~\ref{si:rne-derivation}).  Particles are resampled whenever the effective sample size drops below a threshold. A typical choice for the proposal is the pretrained drift plus a guidance term, $a=b_\theta+\tfrac{\sigma^2}{2}s_\theta+\lambda(t)\,\nabla_{\mathbf{x}} r_t$, when $r_t$ is differentiable; the reward itself, however, need not be, as its gradients enter only through this optional guidance. With $\lambda=0$ the proposal is the pretrained sampler and the scheme reduces to pure reweighting, while $\lambda(t)=\sigma(t)^2/2$ yields classifier-style guidance.

For the intermediate reward we use different constructions: When $r$ is meaningful on noisy intermediate configurations, the annealed reward $r_t=\alpha(t)\,\eta\, r(\mathbf{x}_t)$ is simple and robust. Many material properties, however, are defined only for clean structures (e.g., surrogate-predicted moduli), so the reward exists only at $t=1$. In that case we exploit the two learned heads to form the posterior-mean (Tweedie) estimate \cite{robbins1992empirical,efron2011tweedie} of the final structure implied by an intermediate state, $\hat{\mathbf{x}}_1(\mathbf{x},t)=\mathbb{E}[\mathbf{x}_1\,|\,\mathbf{x}_t=\mathbf{x}]$, which is available in closed form from $b_\theta$ and $s_\theta$ (SI), and set $r_t(\mathbf{x})=\alpha(t)\,\eta\,r\big(\hat{\mathbf{x}}_1(\mathbf{x},t)\big)$, where the annealing schedule suppresses the estimate at small $t$, where it is least reliable. A comparison of intermediate-reward choices and further implementation details are provided in SI~\ref{si:rne}.

\subsection*{Inference-time steering for umbrella sampling and potential of mean force estimation}
Beyond thermodynamics, one is frequently interested in the kinetics or free energy along a chosen reaction coordinate (collective variable) $\xi(\mathbf{x})$, i.e., the potential of mean force (PMF) $F(\xi)=-k_{\mathrm B}T\log\rho(\xi)$ of the induced marginal density $\rho(\xi)=\int p_{\mathrm{target}}(\mathbf{x})\,\delta(\xi(\mathbf{x})-\xi)\,\mathrm{d}\mathbf{x}$. Because $\rho$ typically varies by many orders of magnitude across a barrier, it is poorly resolved by direct sampling. Umbrella sampling \cite{torrie1977nonphysical} circumvents this by stratifying the coordinate into $K$ overlapping windows, each confined near a target value $\xi_k$ by a harmonic bias $u_k(\xi)=\tfrac12\kappa_k(\xi-\xi_k)^2$; every biased ensemble is sampled separately and $\rho$ is reconstructed by reweighting and stitching the window histograms. The key observation is that each biased window is itself a reward-tilted ensemble, $p_k\propto p_{\mathrm{target}}\,e^{-\beta u_k}$, i.e., the construction above with reward $\eta r=-\beta u_k$. We therefore realize the windows by steering the pretrained model at inference time through Eq.~\eqref{eq:m-smc}, with no additional training and no simulation against the interatomic potential. Moreover, the same accumulated weights that draw each window also return its free energy $\log Z_k$, so the windows combine into an absolute PMF $F(\xi)$ by inverse-variance-weighted reweighting. We demonstrate this capability on the KA target, using the separation of a tagged $A$--$B$ pair as the collective variable; the coordinate, the reward construction, and the comparison against the reference profile are detailed in the SI (Section~\ref{si:km-umbrella}).

\subsection*{Model architecture and size-transfer}
Both vector fields are emitted as two heads of a single $E(3)$-equivariant graph neural network with a shared trunk, a torus-adapted PaiNN backbone \cite{schutt2021equivariant}. The network is built to respect the exact symmetries of the target density. Translation invariance is enforced by forming every pairwise feature from minimum-image displacements under periodic boundaries and by projecting both outputs onto the zero-center-of-mass subspace, matching the removed global-translation mode. Rotation and reflection equivariance follows PaiNN's separation into invariant scalar and equivariant vector features. Messages are exchanged on a neighbor graph dynamically rebuilt from the live configuration at every evaluation, which lets the same weights act on the disordered, rearranging structures encountered while transporting the uniform prior. Temperature or composition conditioning enters as additional node embeddings.

Because message passing is local, i.e., each atom interacts only with neighbors inside a finite cutoff, the drift acting on an atom depends only on its immediate environment and is defined for any number of atoms. A model trained on moderate supercells can therefore be evaluated directly on larger cells without retraining, provided the density and the local structural correlations are preserved. This allows for deployment beyond the sizes the model has seen during training, as illustrated in Figure \ref{fig2}. Full architectural details are given in the SI (Section~\ref{si:architecture}).

\subsection*{Target systems and interatomic potentials}
We train ATLAS on five target systems spanning a model glass former and four metallic glasses of increasing chemical complexity. All potentials are evaluated under full periodic boundary conditions. To prevent collapsed configurations during generation, we introduce a short-range soft core using linear continuation below $1.5$\,\AA{} for metallic systems and a logarithmic continuation below $0.8\,\sigma_{AA}$ for the model glass. We note that these modifications do not affect the equilibrium region.

The Kob-Andersen (KA) model is an $80{:}20$ binary Lennard--Jones mixture simulated with $N$ particles on the two-dimensional torus. In reduced units, $\epsilon_{AA}=\sigma_{AA}=k_B=1$, and we set the parameters for KA model as $(\epsilon_{AA},\sigma_{AA})=(1.0,1.0)$, $(\epsilon_{AB},\sigma_{AB})=(1.5,0.8)$ and $(\epsilon_{BB},\sigma_{BB})=(0.5,0.88)$.
Each interaction is shifted to zero at $r_{\mathrm{cut}}^{\alpha\beta}=2.5\,\sigma_{\alpha\beta}$. The square-cell size is set by the reduced number density $\rho=1$.

The Cu-Zr binary glass, Cr-Co-Ni medium-entropy glass and Fe-Ni-Cr-Co-Cu high-entropy glass are described using embedded-atom-method potentials documented in the NIST Interatomic Potentials Repository \cite{hale2018evaluating}. Simulations are performed at $T=700$~K for Cu-Zr and $T=500$~K for the remaining metallic glasses. Each system contains $N=128$ atoms in a fixed cubic cell, with a composition-dependent volume determined from a Vegard-style mixture of elemental atomic volumes. Specifically, the most stable crystal of each element $i$ is relaxed using the MACE-MPA-0 foundation MLIP \cite{batatia2022mace,batatia2025foundation} to obtain its equilibrium per-atom volume $V_i$. For composition $\mathbf{c}=\{c_i\}$, with $\sum_i c_i=1$, we set the atomic volume $V_{\mathrm{atom}}(\mathbf{c})$ and the length of cubic supercell $L(\mathbf{c})$ as
\begin{equation}
V_{\mathrm{atom}}(\mathbf{c})=\sum_i c_iV_i,
\qquad
L(\mathbf{c})=\bigl[NV_{\mathrm{atom}}(\mathbf{c})\bigr]^{1/3}.
\label{eq:densrule}
\end{equation}
This linear mixing rule agrees with direct cell-volume relaxation to within approximately 1\% across the tested configurations.
The eight-element Fe-Mn-Ni-Ti-Cu-Cr-Co-Al system follows the same setup, but its target energy is evaluated directly using the $\sim9$-million-parameter MACE-MPA-0-medium potential. This system tests whether ATLAS can remain efficient when sampling chemically complex alloys with an expensive universal MLIP for which suitable classical potentials may be unavailable.

\subsection*{Evaluation of target properties}
The Warren-Cowley short range order parameters are evaluated according to Eq.\,\ref{SRO}, where the cutoff radius is defined by the first minima of the pair distribution function $g(r)$ of the corresponding paired species.
The mechanical objectives are evaluated with the EAM potentials for Cu-Zr and Cr-Co-Ni, and with the MACE-MPA-0 MLIP for the high-entropy metallic glasses.
For property evaluation, each generated amorphous structure is first fully relaxed. We then obtain the bulk modulus $B$ from an isotropic Birch-Murnaghan equation-of-state fit and the shear modulus $G$ from orthorhombic and monoclinic volume-conserving deformations, combined using the Voigt-Reuss-Hill average. Structures that fail the corresponding Born stability criteria are excluded. The Pugh ratio is subsequently calculated as $B/G$. 

For efficient reward-guided fine-tuning and inference-time tilting of ATLAS with guidance, the target bulk modulus is required to be differentiable with respect to the generated atomic coordinates. 
In our experiments, neural surrogates based on intensive structural descriptors showed negligible predictive power within a fixed composition. We therefore construct a physics-based differentiable surrogate by combining straight-through structural relaxation with a gradient-safe, closed-form affine equation of state. Specifically, each generated configuration $\mathbf{x}$ is first relaxed by gradient descent to its inherent structure $\mathbf{x}^\star$, after which the bulk modulus $B(\mathbf{x}^\star)$ is evaluated from the closed-form equation of state. 
To avoid differentiating through the iterative relaxation trajectory, we use the straight-through approximation $\mathbf{x}^\star=\mathbf{x}+\texttt{stopgrad}\left[\texttt{relax}(\mathbf{x})-\mathbf{x}\right]$, such that $\partial\mathbf{x}^\star/\partial\mathbf{x}=\mathbb{I}$ in the backward pass and gradients are propagated only through the equation-of-state evaluation. 
The resulting reward gradient, $\nabla_{\mathbf{x}}B\approx\nabla_{\mathbf{x}^\star}B$, is therefore an approximate surrogate gradient. This approximation appears sufficient for reward-guided sampling because the guidance requires a direction that increases the target property, rather than the exact total derivative, as demonstrated in Fig.\,\ref{fig3}d.

\subsection*{LLM evolutionary algorithm for active learning}
The main goal of the evolutionary algorithm is to identify optimal compositions with respect to a given reward/objective function $r$ that evaluates on multiple ensemble average properties:
\begin{equation}
\mathbf{x}^* \in \arg\max_{\mathbf{x}} r(\mathbf{x}).
\end{equation}  
The LLM-driven evolutionary algorithm (LLM-EA) acts as an outer-loop controller, using property evaluations from ATLAS to propose new compositions and optimize target properties for materials inverse design. The language model acts as the proposal (evolutionary operator) of an evolutionary workflow operating directly on the population of previously evaluated compositions \cite{wang2025efficient,novikov2025alphaevolve}. Whenever computational resources become available, the controller determines the number of free GPU slots and fills them sequentially, issuing one composition proposal per LLM call.

For each proposal, the controller assembles an in-context set of 10 evaluated compositions that balances exploitation and exploration. In the single-objective setting, e.g. optimizing a positive scalar $y$, this set contains deterministic anchors corresponding to the best-performing compositions, together with stochastic examples sampled from the remaining population according to $\operatorname{softmax}(\tilde{y}/\tau)$, where $\tilde{y}=(y-y_\text{min})/(y_\text{max}-y_\text{min})$ is the min-max-normalized objective and $\tau=0.3$. 
This sampling encourages the LLM to explore a different high-performing subset of the population on each call while retaining the current optima as fixed references. 
The LLM prompt includes the measured compositions and properties, the best result obtained so far, and relevant materials priors, such as approximate elemental moduli. It is also explicitly prompted to avoid duplicate jobs and encourage diversity among concurrently evaluated candidates. We defer the algorithm details to SI \ref{si:llmea}.

For joint optimization of the multiple properties, such as bulk modulus $B$, shear modulus $G$, and Pugh ratio $B/G$ in this work, we scalarize multiple objective values with randomly sampled weights to capture different regions of the Pareto frontier \cite{knowles2006parego}. 
On each proposal call, we uniformly sample non-negative weights, constrained to sum to one, and assign to the objectives. These weights define a temporary weighted combination of the min-max-normalized objectives, which is used only to rank and select the in-context examples.
We note that this random scalarization during sampling is used only to select the in-context examples, and the LLM is always provided with the individual values of $B$, $G$, and $B/G$. 
At each LLM call, the in-context set contains three deterministic anchors: the compositions with the highest measured $B$, $G$, and $B/G$, respectively. The remaining seven examples are sampled using the stochastic scalarization described above. Of these, four are sampled exclusively from the current Pareto frontier, whereas the remaining three are sampled from the dominated population. In both cases, the sampling probabilities are determined by the randomly scalarized scores.
Redrawing the weight vector for every call changes the emphasized trade-off direction and promotes coverage of the full Pareto frontier, while the inclusion of dominated examples provides information about lower-performing regions of composition space. 

Each LLM response is returned as structured JSON and validated against the composition constraints, including non-negativity, unit sum and, for the eight-element search space, support on at most five elements. If more than five elements are proposed, the lowest-concentration components are removed and their fractions redistributed among the retained elements. 
Other invalid or duplicate proposals are rejected and re-requested rather than silently projected onto the feasible domain. The campaign is initialized with a candidate pool sampled uniformly in the composition space, and subsequently proceeds online and asynchronously. 
As simulation jobs finish, their evaluated properties are immediately added to the measured population, and newly available GPU slots are refilled using the updated results, without waiting for the remaining jobs in the current batch to complete. 
In-flight compositions and proposals already issued during the same refill cycle are included in the avoid-list, preventing duplicate evaluations while allowing the controller to manage multiple concurrent ATLAS jobs. 
Acquisition requests are routed through an ordered fallback chain of model deployments so that failure of an individual deployment does not interrupt the optimization loop.

\section*{Acknowledgments}
The authors thank Jutta Rogal and Juno Nam for helpful discussions. 

\clearpage
\begingroup
\sffamily
\raggedright
\setlength{\parindent}{0pt}

\vspace*{-24pt}

{\bfseries\fontsize{18}{22}\selectfont
ATLAS: A Foundation Neural Sampler for Amorphous Materials: Supplementary Information
\par}
\vspace{9pt}
{\bfseries\selectfont
Mouyang Cheng$^{1,2,3,*,\#}$, Denis Blessing$^{1,4,*,\#}$, Botao Yu$^{1,5,\#}$, Gerhard Neumann$^4$,
Mingda Li$^{2,6}$, Carles Domingo-Enrich$^{1,\dagger}$, and Yuanqi Du$^{1,\dagger}$
\par}
\vspace{12pt}
{\fontsize{9.5}{11.5}\selectfont
$^{1}$Microsoft Research New England, Cambridge, MA 02142, USA\\
$^{2}$Center for Computational Science and Engineering, MIT, Cambridge, MA 02139, USA\\
$^{3}$Department of Materials Science and Engineering, MIT, Cambridge, MA 02139, USA\\
$^{4}$Karlsruhe Institute of Technology, 76131 Karlsruhe, Germany\\
$^{5}$Department of Computer Science and Engineering, OSU, Columbus, OH 43210, USA\\
$^{6}$Department of Nuclear Science and Engineering, MIT, Cambridge, MA 02139, USA\\
$^{*}$These authors contributed equally.\\
$^\#$Work done during internship at Microsoft Research.\\
$\dagger$ Correspondence. Email: carlesd@microsoft.com, yuanqidu@microsoft.com
\par}
\endgroup
\tableofcontents
\setcounter{figure}{0}
\setcounter{table}{0}
\setcounter{equation}{0}
\renewcommand{\thefigure}{S\arabic{figure}}
\renewcommand{\thetable}{S\arabic{table}}
\renewcommand{\theequation}{S\arabic{equation}}
\renewcommand{\theHfigure}{Supp\arabic{figure}}
\renewcommand{\theHtable}{Supp\arabic{table}}
\renewcommand{\theHequation}{Supp\arabic{equation}}

\newcommand{\enghist}[3]{%
  \begin{scope}[shift={(#1,#2)}]
    \pgfmathsetmacro{\sig}{0.28-0.18*(#3)}%
    \pgfmathsetmacro{\amp}{0.20+0.28*(#3)}%
    \pgfmathtruncatemacro{\mixp}{100-100*(#3)}%
    \draw[rounded corners=1pt,draw=black!45,line width=0.4pt,fill=white] (-0.40,-0.32) rectangle (0.40,0.32);
    \fill[red!\mixp!blue,opacity=0.20]
      plot[domain=-0.34:0.34,samples=34,smooth] (\x,{-0.24+\amp*exp(-(\x*\x)/(2*\sig*\sig))})
      -- (0.34,-0.24) -- (-0.34,-0.24) -- cycle;
    \draw[red!\mixp!blue,line width=0.8pt]
      plot[domain=-0.34:0.34,samples=34,smooth] (\x,{-0.24+\amp*exp(-(\x*\x)/(2*\sig*\sig))});
  \end{scope}%
}
\newtheorem{proposition}{Proposition}
\newcommand{\daggerfootnote}{\textsuperscript{\textdagger}}

\section{Method details}
Throughout this section we use the conventions of the main-text: generation time runs over $t\in[0,1]$, the prior is the uniform distribution $p_0=\mathrm{Unif}(\mathbb{T}^D)$ on the flat torus $\mathbb{T}^D=(\mathbb{R}/\mathbb{Z})^D$ with $D=dN$, the wrapping map is $\mathrm{wrap}(y)=y-\lfloor y+\tfrac12\rfloor$ (applied componentwise), and $\beta=1/(k_{\mathrm B}T)$. The interpolant, fields, forward/backward SDEs and score target are Eqs.~\eqref{eq:m-interpolant}--\eqref{eq:m-tsi} of the main text. Both $\mathbb{E}[\cdot]$ and $\langle\cdot\rangle$ denote expectations: we reserve $\langle\cdot\rangle$ for thermodynamic averages of physical observables, such as the mean energy $\langle U\rangle$, and write $\mathbb{E}[\cdot]$ for all other expectations.

\subsection{Training algorithm with conditional amortization}
\label{si:algorithm}
Algorithm~\ref{alg:train} states the full training loop, including the conditioning variable $\mathbf{c}$ (main-text Methods). The loop alternates between three operations: \emph{generate} terminal configurations with the current model, \emph{label} them with forces from the interatomic potential, and \emph{regress} the two heads against force-informed targets built from freshly drawn interpolant points. Because the buffer $\mathcal{B}$ stores the (potentially expensive) force evaluations, each evaluation is reused across many updates, making the algorithm highly efficient. When the condition encodes temperature, $\beta=\beta(\mathbf{c})$ enters the score target through the rescaled force $\beta\mathbf{F}$; when it encodes composition at fixed temperature, $\beta$ is constant and $\mathbf{c}$ enters through the species-dependent potential $U(\cdot\,;\mathbf{c})$.

\begin{algorithm}[t]
\caption{\name{} training with conditional amortization}
\label{alg:train}
\begin{algorithmic}[1]
\State \textbf{input:} energy/force $U(\cdot\,;\mathbf{c}),\ \mathbf{F}(\cdot\,;\mathbf{c})=-\nabla U(\cdot\,;\mathbf{c})$; condition prior $p(\mathbf{c})$; schedules $\alpha(t),\gamma(t),\sigma(t)$; prior $p_0=\mathrm{Unif}(\mathbb{T}^D)$
\State initialize heads $(b_\theta,s_\theta)$; replay buffer $\mathcal{B}\gets\emptyset$
\For{a fixed number of outer steps}
  \State \textbf{// generate}
  \State draw $\mathbf{c}\sim p(\mathbf{c})$ and $\mathbf{x}_0\sim p_0$
  \State integrate the forward SDE~\eqref{eq:m-fwd} with drift $b_\theta+\tfrac{\sigma^2}{2}s_\theta$ under condition $\mathbf{c}$ to obtain $\mathbf{x}_1\sim p_{\mathrm{model}}(\cdot\mid\mathbf{c})$
  \State \textbf{// label}
  \State evaluate forces $\mathbf{F}(\mathbf{x}_1;\mathbf{c})$; push $(\mathbf{F}(\mathbf{x}_1;\mathbf{c}),\mathbf{x}_1,\mathbf{c})$ onto $\mathcal{B}$
  \For{a fixed number of inner steps}
    \State draw $(\mathbf{F}(\mathbf{x}_1;\mathbf{c}),\mathbf{x}_1,\mathbf{c})\sim\mathcal{B}$, $\mathbf{x}_0\sim p_0$, $\mathbf{z}\sim\mathcal{N}(\mathbf{0},I)$, $t\sim\mathcal{U}[0,1]$
  \State \textbf{// interpolate}
    \State $\boldsymbol{\Delta}\gets\mathrm{wrap}(\mathbf{x}_1-\mathbf{x}_0)$;\quad $\mathbf{x}_t\gets\mathrm{wrap}\!\big(\mathbf{x}_0+\alpha(t)\boldsymbol{\Delta}+\gamma(t)\mathbf{z}\big)$
    \State $\hat{\mathbf{s}}\gets\beta(\mathbf{c})\,\mathbf{F}(\mathbf{x}_1;\mathbf{c})$;\quad $\hat{\mathbf{b}}\gets\dot\alpha(t)\boldsymbol{\Delta}-\dot\gamma(t)\gamma(t)\,\hat{\mathbf{s}}$
  \State \textbf{// regress}
    \State $\ell(\theta)\gets\big\|s_\theta(\mathbf{x}_t,t,\mathbf{c})-\hat{\mathbf{s}}\big\|^2+\big\|b_\theta(\mathbf{x}_t,t,\mathbf{c})-\hat{\mathbf{b}}\big\|^2$
    \State take a gradient step on $\ell(\theta)$
  \EndFor
\EndFor
\State \textbf{return} $(b_\theta,s_\theta)$
\end{algorithmic}
\end{algorithm}

\subsection{The target score identity on the torus}
\label{si:tsi}
We prove the identity~\eqref{eq:m-tsi} used to train \name{}: under a uniform prior, the interatomic force at the terminal configuration is an unbiased estimator of the interpolant score.

\begin{proposition}[Target score identity]
\label{prop:tsi}
Let $p_0=\mathrm{Unif}(\mathbb{T}^D)$, let $p_{\mathrm{target}}$ be $C^1$ and strictly positive on $\mathbb{T}^D$, and let $\mathbf{x}_t$ be the interpolant~\eqref{eq:m-interpolant} with $\mathbf{x}_1\sim p_{\mathrm{target}}$. Then for every $t\in(0,1)$ and every $\mathbf{x}$,
\begin{equation}
    s(\mathbf{x},t)=\nabla_{\mathbf{x}}\log p_t(\mathbf{x})=\mathbb{E}\big[\beta\mathbf{F}(\mathbf{x}_1)\,\big|\,\mathbf{x}_t=\mathbf{x}\big].
    \label{eq:s-tsi}
\end{equation}
No smoothness of the interpolant kernel is required; in particular the identity holds for the deterministic interpolant $\gamma\equiv0$.
\end{proposition}

\begin{proof}
Since $\boldsymbol{\Delta}=\mathrm{wrap}(\mathbf{x}_1-\mathbf{x}_0)\equiv\mathbf{x}_1-\mathbf{x}_0\ (\mathrm{mod}\ 1)$, the bridge mean satisfies $\mathbf{x}_0+\alpha(t)\boldsymbol{\Delta}\equiv\mathbf{x}_1-(1-\alpha(t))\boldsymbol{\Delta}\ (\mathrm{mod}\ 1)$, so we may write the interpolant from the target endpoint,
\begin{equation}
    \mathbf{x}_t=\mathrm{wrap}(\mathbf{x}_1+\boldsymbol{\xi}_t),
    \qquad
    \boldsymbol{\xi}_t=-(1-\alpha(t))\boldsymbol{\Delta}+\gamma(t)\mathbf{z}.
    \label{eq:s-xi}
\end{equation}
Because $\mathbf{x}_0$ is uniform on the torus, $\boldsymbol{\Delta}=\mathrm{wrap}(\mathbf{x}_1-\mathbf{x}_0)$ is uniform and independent of $\mathbf{x}_1$; hence $\boldsymbol{\xi}_t$ is independent of $\mathbf{x}_1$, and the conditional law of $\mathbf{x}_t$ given $\mathbf{x}_1$ is a fixed translation kernel $\eta_t(\cdot-\mathbf{x}_1)$, where $\eta_t$ denotes the density of $\boldsymbol{\xi}_t$ on $\mathbb{T}^D$. The marginal is therefore the circular convolution
\begin{equation}
    p_t(\mathbf{x})=\int_{\mathbb{T}^D}p_{\mathrm{target}}(\mathbf{y})\,\eta_t(\mathbf{x}-\mathbf{y})\,d\mathbf{y},
    \label{eq:s-conv}
\end{equation}
where $\eta_t$ is a probability density on the torus: smooth and strictly positive when $\gamma(t)>0$, and, for the deterministic interpolant $\gamma\equiv0$, the uniform density on the wrapped cube of side length $1-\alpha(t)$. In either case $p_t=p_{\mathrm{target}}\ast\eta_t$ inherits the regularity and strict positivity of $p_{\mathrm{target}}$, with $p_t\geq\min p_{\mathrm{target}}>0$. Substituting $\mathbf{u}=\mathbf{x}-\mathbf{y}$ in~\eqref{eq:s-conv} and differentiating under the integral---the gradient falls on the smooth endpoint density, never on the kernel---
\begin{equation}
    \nabla_{\mathbf{x}}p_t(\mathbf{x})
    =\int_{\mathbb{T}^D}\nabla p_{\mathrm{target}}(\mathbf{x}-\mathbf{u})\,\eta_t(\mathbf{u})\,d\mathbf{u}
    =\int_{\mathbb{T}^D}p_{\mathrm{target}}(\mathbf{y})\,\beta\mathbf{F}(\mathbf{y})\,\eta_t(\mathbf{x}-\mathbf{y})\,d\mathbf{y},
    \label{eq:s-gradconv}
\end{equation}
using $\nabla p_{\mathrm{target}}=p_{\mathrm{target}}\nabla\log p_{\mathrm{target}}=p_{\mathrm{target}}\,\beta\mathbf{F}$; differentiation under the integral is justified by dominated convergence, since $\nabla p_{\mathrm{target}}$ is bounded on the compact torus and $\eta_t$ is a probability density. By Bayes' rule the posterior density of $\mathbf{x}_1=\mathbf{y}$ given $\mathbf{x}_t=\mathbf{x}$ is $p_{\mathrm{target}}(\mathbf{y})\,\eta_t(\mathbf{x}-\mathbf{y})/p_t(\mathbf{x})$, so dividing~\eqref{eq:s-gradconv} by $p_t(\mathbf{x})$ turns the right-hand side into the posterior average of $\beta\mathbf{F}(\mathbf{x}_1)$, which is~\eqref{eq:s-tsi}.
\end{proof}

\subsection{Uniqueness of the fixed-point}
\label{si: uniqueness}
We now justify the claim made below Eq.~\eqref{eq:m-loss}, i.e., $p_{\mathrm{target}}$ is the \emph{only} distribution left invariant by the fixed-point round of Fig.~\ref{fig:m-fixedpoint}, and we identify the sharp condition on the schedules under which this holds. We work in the idealized regime in which each head attains its regression optimum (unbounded capacity and exact expectations). Throughout, all densities are of class $C^2$ and strictly positive on $\mathbb{T}^D$; $\alpha$ is $C^1$ and strictly increasing with $\alpha(0)=0$, $\alpha(1)=1$; $\gamma$ is $C^1$ with $\gamma(0)=\gamma(1)=0$; and $\sigma$ is continuous. Under these assumptions the identities of Section~\ref{si:tsi} apply; as shown there, they require no smoothness of the interpolant kernel and hence cover the deterministic interpolant $\gamma\equiv0$ used in our experiments.

\paragraph{The training round as a map on distributions.}
Write $q=p_{\mathrm{model}}$ for the model's terminal law at the start of a round. The round generates endpoints $\mathbf{x}_1\sim q$, labels them with the rescaled force $\beta\mathbf{F}(\mathbf{x}_1)=\nabla\log p_{\mathrm{target}}(\mathbf{x}_1)$, and minimizes~\eqref{eq:m-loss} over interpolants~\eqref{eq:m-interpolant} built toward these endpoints. Because a least-squares fit returns the conditional expectation of its target, the two heads become
\begin{equation}
    s'(\mathbf{x},t)=\mathbb{E}_q\big[\nabla\log p_{\mathrm{target}}(\mathbf{x}_1)\,\big|\,\mathbf{x}_t=\mathbf{x}\big],
    \qquad
    b'(\mathbf{x},t)=\mathbb{E}_q\big[\dot\alpha(t)\boldsymbol{\Delta}-\dot\gamma(t)\gamma(t)\,\nabla\log p_{\mathrm{target}}(\mathbf{x}_1)\,\big|\,\mathbf{x}_t=\mathbf{x}\big],
    \label{eq:s-round}
\end{equation}
where the subscript $q$ records that the interpolant endpoint is drawn from the current model. The round then \emph{generates} with these heads: the next buffer is produced by integrating the forward SDE with drift $v'=b'+\tfrac{\sigma^2}{2}s'$ from the uniform prior; we write $\rho_t$ for its marginals, so $\rho_1$ is the next endpoint law.

Two exact identities organize the analysis. First, the regressed velocity in~\eqref{eq:s-round} and the true velocity $b^q(\mathbf{x},t)=\mathbb{E}_q[\dot\alpha(t)\boldsymbol{\Delta}+\dot\gamma(t)\mathbf{z}\mid\mathbf{x}_t=\mathbf{x}]$ of Eq.~\eqref{eq:m-bs} share the term $\dot\alpha(t)\,\mathbb{E}_q[\boldsymbol{\Delta}\mid\mathbf{x}_t]$ and differ only through the score, $b'-b^q=-\dot\gamma(t)\gamma(t)\,(s'-s^q)$ (using $\mathbb{E}_q[\mathbf{z}\mid\mathbf{x}_t]=-\gamma(t)\,s^q$ from~\eqref{eq:m-bs}), where $s^q(\mathbf{x},t)=\nabla\log p^q_t(\mathbf{x})$ is the score of the interpolant marginal $p^q_t=q\ast\eta_t$ (the convolution~\eqref{eq:s-conv} with $p_{\mathrm{target}}$ replaced by $q$). For the forward sampling drifts $v'=b'+\tfrac{\sigma^2}{2}s'$ and $v^q=b^q+\tfrac{\sigma^2}{2}s^q$ this gives
\begin{equation}
    v'(\mathbf{x},t)-v^{q}(\mathbf{x},t)=c(t)\,\big(s'(\mathbf{x},t)-s^{q}(\mathbf{x},t)\big),
    \qquad
    c(t):=\tfrac{1}{2}\sigma(t)^2-\dot\gamma(t)\gamma(t),
    \label{eq:s-drift}
\end{equation}
so the force labels influence generation \emph{only} through the score-consumption coefficient $c$. Second, applying the target-score identity (Proposition~\ref{prop:tsi}) with the endpoint law taken to be $q$ writes the true score of $q$ as the posterior average of its \emph{own} force,
\begin{equation}
    s^q(\mathbf{x},t)=\mathbb{E}_q\big[\nabla\log q(\mathbf{x}_1)\,\big|\,\mathbf{x}_t=\mathbf{x}\big],
    \label{eq:s-selfscore}
\end{equation}
so that, with $\phi:=\log\big(p_{\mathrm{target}}/q\big)$ and the Bayes representation of the endpoint posterior (Section~\ref{si:tsi}),
\begin{equation}
    s'(\mathbf{x},t)-s^{q}(\mathbf{x},t)
    =\mathbb{E}_q\big[\nabla\phi(\mathbf{x}_1)\,\big|\,\mathbf{x}_t=\mathbf{x}\big]
    =\frac{\big[(q\,\nabla\phi)\ast\eta_t\big](\mathbf{x})}{p^q_t(\mathbf{x})}.
    \label{eq:s-error}
\end{equation}
The decisive feature of~\eqref{eq:s-round} is that the \emph{label} transports the target force $\nabla\log p_{\mathrm{target}}$, whereas the conditional average runs over the model's own configurations; the error field~\eqref{eq:s-error} vanishes identically if and only if $\phi$ is constant, i.e.\ $q=p_{\mathrm{target}}$.

We call $q$ a \emph{fixed point of the round} when the sampler reproduces the interpolant marginals it was trained on,
\begin{equation}
    \rho_t=p^q_t \qquad \text{for all } t\in[0,1].
    \label{eq:s-fixdef}
\end{equation}
This is the natural notion of stationarity for the model as a whole: it implies in particular $\rho_1=q$, so the buffer, the regression targets~\eqref{eq:s-round}, and hence the fields are all unchanged by the round and training is stationary; and it is exactly the requirement that the two heads remain, at every time, the velocity and score of the sampler's own marginals---the interpretation on which the corrector and churn samplers (Section~\ref{si:pc}), the Tweedie estimate (Section~\ref{si:tweedie}) and the backward kernels of the path-weight estimator rely.

\begin{figure}[tbp!]
 \centering
 \includegraphics[width=0.95\textwidth]{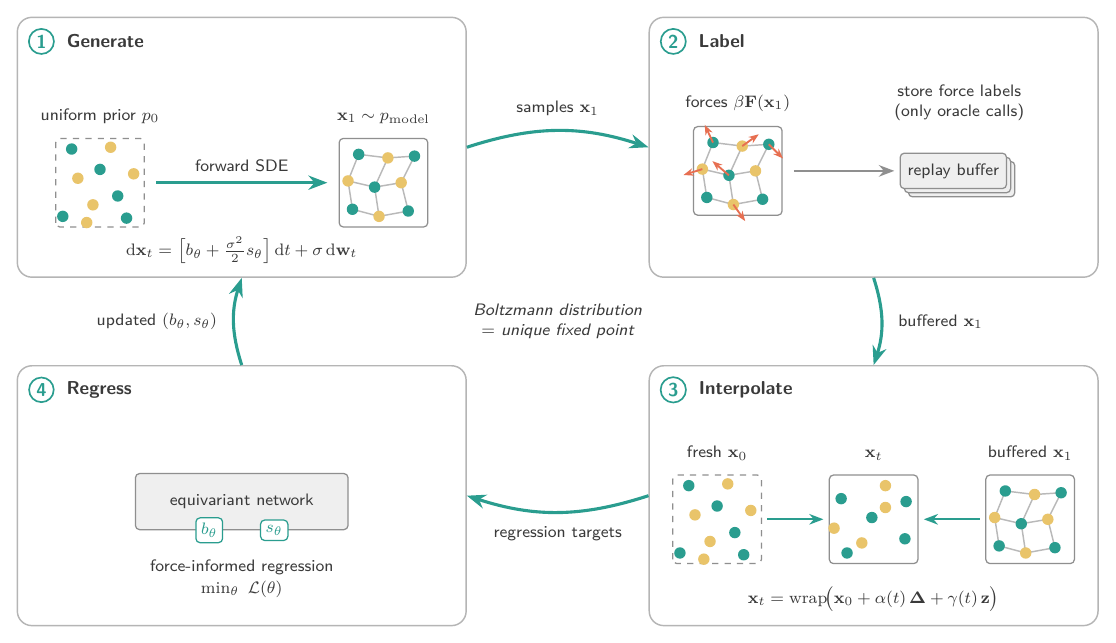}
 \caption{\textbf{Self-consistent fixed-point training of ATLAS.} At each round, the current model generates terminal configurations by integrating the forward SDE from the uniform prior (\emph{generate}); the interatomic potential labels them with forces, which are stored in a replay buffer (\emph{label}); intermediate states $\mathbf{x}_t$ are drawn from the interpolant connecting fresh prior samples to buffered configurations (\emph{interpolate}); and the two network heads $(b_\theta,s_\theta)$ are updated by force-informed regression (\emph{regress}). The Boltzmann distribution is the unique distribution for which this cycle is stationary.}
 \label{fig:m-fixedpoint}
\end{figure}

\begin{proposition}[Unique fixed point]
\label{prop:uniqueness}
Under the standing assumptions, with $c$ as in~\eqref{eq:s-drift}:
\begin{enumerate}
    \item[(i)] If $c\not\equiv0$ on $(0,1)$, then $q$ is a fixed point of the round if and only if $q=p_{\mathrm{target}}$.
    \item[(ii)] If $c\equiv0$ on $(0,1)$, then \emph{every} admissible $q$ is a fixed point: the iteration is degenerate and the force labels never influence generation.
\end{enumerate}
\end{proposition}

\begin{proof}
\emph{(ii) and sufficiency in (i).} If $c\equiv0$, then $v'=v^{q}$ by~\eqref{eq:s-drift} for every $q$; the sampler is the true forward SDE of $q$'s interpolant, whose marginals are $p^q_t$ by~\eqref{eq:m-fwd}, so~\eqref{eq:s-fixdef} holds for every $q$. If instead $q=p_{\mathrm{target}}$, then $\phi\equiv0$, the error field~\eqref{eq:s-error} vanishes, both heads coincide with the true fields of $p_{\mathrm{target}}$, and again the sampler reproduces the interpolant marginals; $p_{\mathrm{target}}$ is a fixed point for every choice of $c$.

\emph{Necessity in (i).} Let $q$ satisfy~\eqref{eq:s-fixdef}. The sampler marginals and the interpolant marginals solve the Fokker--Planck equations
\begin{equation*}
    \partial_t\rho=-\nabla\!\cdot\!\big(\rho\,v'\big)+\tfrac{\sigma^2}{2}\Delta\rho,
    \qquad
    \partial_t p^q=-\nabla\!\cdot\!\big(p^q\,v^{q}\big)+\tfrac{\sigma^2}{2}\Delta p^q,
\end{equation*}
with the same initial datum $p_0$; under the standing smoothness all coefficients are continuous and the subtraction below is classical. Substituting $\rho_t=p^q_t$ into the first equation and subtracting the second, the transport and diffusion terms cancel and, by~\eqref{eq:s-drift} and~\eqref{eq:s-error},
\begin{equation}
    0=\nabla\!\cdot\!\Big(p^q_t\,c(t)\big(s'-s^{q}\big)\Big)
    =c(t)\,\nabla\!\cdot\!\big[(q\,\nabla\phi)\ast\eta_t\big]
    =c(t)\,\big[\nabla\!\cdot\!(q\,\nabla\phi)\big]\ast\eta_t,
    \qquad t\in(0,1),
    \label{eq:s-divfree}
\end{equation}
where the last equality moves the divergence through the convolution ($q\,\nabla\phi\in C^1$). Since $c$ is continuous and $c\not\equiv0$, there is an open interval $I\subset(0,1)$ with $c\neq0$ on $I$. Because $\alpha$ is continuous and strictly increasing, the set $\{t\in I:\ 1-\alpha(t)\in\mathbb{Q}\}$ is countable; choose $t_\star\in I$ with $\theta:=1-\alpha(t_\star)$ irrational.
Let $Q=[-\tfrac12,\tfrac12)^D$ and let
$\pi:\mathbb R^D\to\mathbb T^D$ denote wrapping. At $t_\star$, the
interpolant displacement relative to its terminal point is
\[
\boldsymbol{\xi}_{t_\star}
=-\theta\boldsymbol{\Delta}
+\gamma(t_\star)\mathbf z,
\]
where $\boldsymbol{\Delta}\sim\mathrm{Unif}(Q)$ and
$\mathbf z\sim\mathcal N(\mathbf0,I)$ are independent. The corresponding
torus convolution kernel is
\[
\eta_{t_\star}
:=\pi_{\#}\operatorname{Law}(\boldsymbol{\xi}_{t_\star}).
\]
For $\mathbf k\in\mathbb Z^D$, wrapping does not change a torus
character because
$\pi(\mathbf u)-\mathbf u\in\mathbb Z^D$. Therefore,
\begin{align}
\widehat{\eta}_{t_\star}(\mathbf k)
&=\mathbb E\!\left[
e^{-2\pi i\mathbf k\cdot\pi(\boldsymbol{\xi}_{t_\star})}
\right] \notag=\mathbb E\!\left[
e^{-2\pi i\mathbf k\cdot\boldsymbol{\xi}_{t_\star}}
\right] 
=e^{-2\pi^2\gamma(t_\star)^2\|\mathbf k\|^2}
  \prod_{j=1}^D
  \int_{-1/2}^{1/2}e^{2\pi i k_j\theta u}\,du 
\notag\\&=e^{-2\pi^2\gamma(t_\star)^2\|\mathbf k\|^2}
  \prod_{j:k_j\neq0}
  \frac{\sin(\pi k_j\theta)}{\pi k_j\theta}.
\label{eq:s-fourier}
\end{align}
A factor vanishes only if $k_j\theta\in\mathbb{Z}\setminus\{0\}$ for some $j$, which is impossible for irrational $\theta$; hence $\widehat{\eta}_{t_\star}(\mathbf{k})\neq0$ for every $\mathbf{k}$. Evaluating~\eqref{eq:s-divfree} at $t_\star$ and passing to Fourier coefficients, $\widehat{f}(\mathbf{k})\,\widehat{\eta}_{t_\star}(\mathbf{k})=0$ for all $\mathbf{k}$, where $f:=\nabla\!\cdot\!(q\,\nabla\phi)$ is continuous; therefore $f\equiv0$ on $\mathbb{T}^D$. Multiplying by $\phi$ and integrating by parts on the closed torus,
\begin{equation*}
    0=\int_{\mathbb{T}^D}\phi\,\nabla\!\cdot\!\big(q\,\nabla\phi\big)\,d\mathbf{x}
    =-\int_{\mathbb{T}^D}q\,\big|\nabla\phi\big|^2\,d\mathbf{x},
\end{equation*}
and since $q>0$ this forces $\nabla\phi\equiv\mathbf{0}$. On the connected torus $\phi$ is constant, so $q=\kappa\,p_{\mathrm{target}}$ with $\kappa>0$, and normalization gives $\kappa=1$, hence $q=p_{\mathrm{target}}$.
\end{proof}

\subsection{Transition kernels on the torus}
\label{app:torus-kernels}
The Euler--Maruyama discretization of the forward SDE \eqref{eq:m-fwd} on the flat torus $\mathbb{T}^d = [0,1)^d$ has transition densities given by \emph{wrapped} Gaussians: writing $\boldsymbol{\mu}_k(\mathbf{x}_k) = \mathbf{x}_k+\big[b_\theta+\tfrac{\sigma^2}{2}s_\theta\big](\mathbf{x}_k,t_k)\,\Delta t_k$ and $\tau_k^2 = \sigma(t_k)^2\Delta t_k$, the exact kernel is
\begin{equation}
    \overrightarrow{p}_k(\mathbf{x}_{k+1}\,|\,\mathbf{x}_k)
    = \mathcal{N}_{\mathbb{T}}\big(\mathbf{x}_{k+1};\,\boldsymbol{\mu}_k,\,\tau_k^2 I\big)
    \coloneqq \sum_{\mathbf{z}\in\mathbb{Z}^d}\mathcal{N}\big(\mathbf{x}_{k+1}+\mathbf{z};\,\boldsymbol{\mu}_k,\,\tau_k^2 I\big),
    \label{eq:app-wrapped}
\end{equation}
i.e., the Euclidean Gaussian density summed over all periodic images of the endpoint, and analogously for the backward kernel $\overleftarrow{p}_k$. The infinite sum in \eqref{eq:app-wrapped} is required for exactness but is dominated by a single term whenever the step standard deviation is small compared to the cell length. Concretely, let
\begin{equation}
    \boldsymbol{\delta}_k \coloneqq \operatorname{wrap}\big(\mathbf{x}_{k+1}-\boldsymbol{\mu}_k\big) \in [-\tfrac12,\tfrac12)^d
\end{equation}
denote the minimum-image displacement, i.e., the representative of $\mathbf{x}_{k+1}-\boldsymbol{\mu}_k$ modulo $\mathbb{Z}^d$ with smallest norm. Retaining only the nearest image yields the \emph{minimum-image approximation}
\begin{equation}
    \log \overrightarrow{p}_k(\mathbf{x}_{k+1}\,|\,\mathbf{x}_k)
    \approx -\frac{\|\boldsymbol{\delta}_k\|^2}{2\tau_k^2} - \frac{d}{2}\log\big(2\pi\tau_k^2\big),
    \label{eq:app-min-image}
\end{equation}
which corresponds to selecting the image $\mathbf{z}^\star$ that maximizes the summand in \eqref{eq:app-wrapped}.

\subsection{Importance-weighted training} \label{SI: IW training}
A limitation of the training objective introduced above is that it relies exclusively on the force $-\nabla E$ and never evaluates the energy itself. Since forces only carry local information about the energy landscape, the algorithm can be blind to the relative masses of metastable basins that are separated by large energy barriers, and the learned sampler may misallocate probability mass across modes \cite{grenioux2026diffusiveclassificationlosslearning}.
 
As a remedy, we incorporate the path weights into the training objective,
\begin{equation}
    \mathcal{L}_{\mathrm{IW}}(\theta)=\mathbb{E}\Big[ \bar w(\mathbf{x}_{0:K})\big(\big\|s_\theta(\mathbf{x}_t,t)-\hat{\mathbf{s}}_t\big\|^2+\big\|b_\theta(\mathbf{x}_t,t)-\hat{\mathbf{b}}_t\big\|^2\big)\Big],
    \label{eq:iw-loss}
\end{equation}
with self-normalized weights $\bar w(\mathbf{x}_{0:K}) =  w(\mathbf{x}_{0:K})/ \mathbb{E}[ w(\mathbf{x}_{0:K})]$, where $w$ is defined as in \eqref{eq:m-logw} and the expectation is approximated by the empirical mean over the batch. The weights reweight trajectories according to the Boltzmann target density and thereby reintroduce the global energy information that the force-based objective lacks.
 
In practice, however, $\bar w$ can suffer from high variance, in particular early in training when the model distribution is still far from the target, causing a small number of trajectories to dominate the batch and destabilizing optimization. To mitigate this and enable stable training, we smooth the importance weights by tempering them with an additional parameter $\eta \in [0,1]$,
\begin{equation}
    \bar w_\eta(\mathbf{x}_{0:K}) = \frac{w(\mathbf{x}_{0:K})^{\eta}}{\mathbb{E}\big[w(\mathbf{x}_{0:K})^{\eta}\big]},
    \label{eq:tempered-w}
\end{equation}
which recovers the unsmoothed weights for $\eta = 1$ and uniform weights for $\eta = 0$. The parameter $\eta$ is chosen adaptively per batch such that the weights retain a prescribed effective sample size (ESS). Concretely, for a batch of $N$ trajectories with weights $w_i = w(\mathbf{x}^{(i)}_{0:K})$, the normalized ESS of the tempered weights is given by
\begin{equation}
    \operatorname{ESS}(\eta) = \frac{\big(\sum_{i=1}^N w_i^{\eta}\big)^2}{\sum_{i=1}^N w_i^{2\eta}} \in [1, N],
    \label{eq:ess}
\end{equation}
which is continuous and monotonically non-increasing in $\eta$, with $\operatorname{ESS}(0) = N$. Given a target fraction $\kappa \in (0,1)$, we set
\begin{equation}
    \eta^\star = \max\big\{\eta \in [0,1] \,:\, \operatorname{ESS}(\eta) \ge \kappa N\big\},
    \label{eq:eta-star}
\end{equation}
that is, $\eta^\star = 1$ whenever the raw weights already satisfy the ESS constraint, and otherwise $\eta^\star$ is obtained by solving $\operatorname{ESS}(\eta) = \kappa N$ via a line search, e.g., bisection, which is cheap and robust due to the monotonicity of \eqref{eq:ess}.
 
Tempering trades variance for bias: the smoothed estimator no longer targets the exact importance-weighted objective but its counterpart under the geometric interpolation between model and target induced by $\eta^\star$. This trade-off is benign in our setting, as the bias reduces as training progresses and the model distribution approaches the target, the weights concentrate around unity, $\operatorname{ESS}(1) \to N$, and hence $\eta^\star \to 1$, recovering the unbiased objective \eqref{eq:iw-loss}.

\subsection{Tweedie's identity for stochastic interpolants}
\label{si:tweedie}
Inference-time steering with terminal-only rewards needs the posterior mean $\hat{\mathbf{x}}_1(\mathbf{x},t)=\mathbb{E}[\mathbf{x}_1\mid\mathbf{x}_t=\mathbf{x}]$ of the clean configuration implied by an intermediate state. For a stochastic interpolant $\mathbf{x}_t=\alpha_t\mathbf{x}_1+\beta_t\mathbf{x}_0+\gamma_t\mathbf{z}$ with $\mathbf{z}\sim\mathcal{N}(\mathbf{0},I)$ independent of the endpoints, this is available from the two learned heads alone, with no additional network. Tweedie's identity gives the conditional noise $\mathbb{E}[\mathbf{z}\mid\mathbf{x}_t]=-\gamma_t\,s(\mathbf{x}_t,t)$; taking $\mathbb{E}[\,\cdot\mid\mathbf{x}_t]$ of $\mathbf{x}_t$ and of its time derivative $\dot{\mathbf{x}}_t=\dot\alpha_t\mathbf{x}_1+\dot\beta_t\mathbf{x}_0+\dot\gamma_t\mathbf{z}$ (whose conditional mean is the velocity $b$) yields a $2\times2$ linear system in $\mathbb{E}[\mathbf{x}_1\mid\mathbf{x}_t]$ and $\mathbb{E}[\mathbf{x}_0\mid\mathbf{x}_t]$. Eliminating the latter,
\begin{equation}
    \mathbb{E}[\mathbf{x}_1\mid\mathbf{x}_t]
    =\frac{\beta_t\,b(\mathbf{x}_t,t)-\dot\beta_t\,\mathbf{x}_t+\gamma_t\big(\beta_t\dot\gamma_t-\dot\beta_t\gamma_t\big)\,s(\mathbf{x}_t,t)}{\beta_t\dot\alpha_t-\alpha_t\dot\beta_t}.
    \label{eq:s-tweedie-gen}
\end{equation}
On the torus, where the interpolant is built on the universal cover $\bar{\mathbf{x}}_t\in\mathbb{R}^d$ with $\mathbf{x}_t=\operatorname{wrap}(\bar{\mathbf{x}}_t)$, the identities for the score and velocity remain exact, but the first equation of the system becomes $\mathbb{E}[\bar{\mathbf{x}}_t\mid\mathbf{x}_t]=\mathbf{x}_t+\mathbf{m}_t(\mathbf{x}_t)$, with $\mathbf{m}_t(\mathbf{x})=\mathbb{E}[\bar{\mathbf{x}}_t-\mathbf{x}_t\mid\mathbf{x}_t=\mathbf{x}]$ the expected winding number of the lift. The exact torus identity is thus \eqref{eq:s-tweedie-gen} with $\mathbf{x}_t$ replaced by $\mathbf{x}_t+\mathbf{m}_t(\mathbf{x}_t)$, and it characterizes the mean lifted endpoint $\mathbb{E}[\bar{\mathbf{x}}_1\mid\mathbf{x}_t]$. Since $\mathbf{m}_t$ is not identifiable from the score and velocity, we heuristically drop it. Note that only the fractional part of $\mathbf{m}_t$ is left out, as its integer part is absorbed by the final wrap; this fractional part vanishes whenever the lifted density is well localized, in which case a single image dominates, $\mathbf{m}_t$ is (close to) integer, and the heuristic is exact. For the \name{} interpolant~\eqref{eq:m-interpolant}, $\beta_t=1-\alpha_t$ and the determinant is constant, $\beta_t\dot\alpha_t-\alpha_t\dot\beta_t=\dot\alpha_t$, giving the closed form used in practice,
\begin{equation}
    \hat{\mathbf{x}}_1(\mathbf{x},t)=\mathrm{wrap}\!\left[\mathbf{x}+\frac{1-\alpha(t)}{\dot\alpha(t)}\,b_\theta(\mathbf{x},t)
    +\left(\frac{(1-\alpha(t))\,\gamma(t)\dot\gamma(t)}{\dot\alpha(t)}+\gamma(t)^2\right)s_\theta(\mathbf{x},t)\right].
    \label{eq:s-tweedie}
\end{equation}

\subsection{Predictor--corrector sampling and churn}
\label{si:pc}
Integrating the forward SDE~\eqref{eq:m-fwd} with a finite step size and an imperfect learned score leaves the generated configurations slightly out of equilibrium: for the KA glass, plain Euler--Maruyama generation lands well above the reference mean energy, with a heavy high-energy tail. Two related refinements recover this gap by spending additional network evaluations on marginal-preserving corrections, without additional force evaluations.

\paragraph{Predictor--corrector.}
The corrector is driven by the learned score head $s_\theta(\mathbf{x},t)\approx\nabla\log p_t$ of the intermediate marginal. A predictor--corrector scheme \cite{song2020score} interleaves generation with in-place relaxation: after each Euler--Maruyama \emph{predictor} step $\mathbf{x}\leftarrow\mathrm{wrap}(\mathbf{x}+(b_\theta+\tfrac{\sigma^2}{2}s_\theta)\,\Delta t+\sigma\sqrt{\Delta t}\,\boldsymbol{\xi})$ that advances $t$, it applies several \emph{corrector} steps of overdamped Langevin dynamics driven by the learned score at the \emph{current} time,
\begin{equation}
    \mathbf{x}\leftarrow\mathrm{wrap}\!\Big(\mathbf{x}+\tfrac{h}{2}\,s_\theta(\mathbf{x},t)+\sqrt{h}\,\boldsymbol{\xi}\Big),\qquad \boldsymbol{\xi}\sim\mathcal{N}(\mathbf{0},I)\ \text{(zero-COM)},
    \label{eq:s-corrector}
\end{equation}
which leave the marginal $p_t$ invariant (their stationary distribution is $p_t$ when $s_\theta=\nabla\log p_t$) and pull the population back onto it without moving $t$. Because no closed-form density is available at intermediate $t$, the corrector step size is set by the signal-to-noise heuristic\cite{song2020score} $h(t)=2\big(r\,\|\boldsymbol{\xi}\|/\|s_\theta(\mathbf{x},t)\|\big)^2$, which can span several orders of magnitude across the noise schedule. The goal of this scheme is to keep the population on $p_t$ at every step so errors are repaired immediately rather than compounding into the wrong basin populations at $t=1$.

\paragraph{Churn.}
Predictor--corrector spends its correction budget in fixed blocks at each rung; \emph{churn} spreads the same budget along the trajectory by injecting extra score-paired noise into a single SDE. Assuming that $b_\theta$ and $s_\theta$ are perfectly learned, for any $\gamma(t)\ge0$ the family $\mathrm{d}\mathbf{x}=\big(b_\theta+\tfrac{\gamma}{2}s_\theta\big)\mathrm{d}t+\sqrt{\gamma}\,\mathrm{d}\mathbf{w}$ shares the marginals $p_t$ ($\gamma=\sigma^2$ is the trained SDE, $\gamma=0$ the probability-flow ODE), and SNR-matched churn takes $\gamma(t)\,\mathrm{d}t=h(t)$ per step from the same heuristic, so each step is one transport increment plus one corrector-strength kick at a single network evaluation. In practice we fix $r=0.16$ and scale the kick by a temperature-dependent strength $\lambda$ (below). We use these samplers to draw the high-fidelity configurations reported for structural and energetic observables reported in Figure \ref{fig2}, and all samples entering the free energy estimates of Fig.~\ref{fig2}b,c.

\paragraph{Churn as used for the KA results.} Each of the $250$ steps applies one predictor step followed by one corrector step, $\mathbf{x}\leftarrow\mathrm{wrap}\big(\mathbf{x}+b_\theta\,\Delta t+\tfrac{h}{2}\,\tilde s_\theta+\sqrt{h}\,\boldsymbol{\xi}\big)$, with $h=\lambda(t)\min\!\big(2(r\|\boldsymbol{\xi}\|/\|\tilde s_\theta\|)^2,\,h_{\max}\big)$, $r=0.16$ and $h_{\max}=0.05$. The strength is $\lambda(t)=\lambda(T)$ for $t\le0.8$ and decays geometrically to $\min(\lambda(T),4)$ at $t=1$. Two modifications of the score enter $\tilde s_\theta$. First, a temperature-dependent gain, $\tilde s_\theta=g(T)\,s_\theta$ with $g(T)=\max(1,1/T)$. Second, the score is clipped to a maximum norm over the whole configuration, $750\sqrt{N/64}$ (the $\sqrt{N}$ scaling keeps the clip per particle independent of the system size). The strength follows the analytic law
\begin{equation}
\lambda(T)=\min\!\Big[\lambda_{\min}\,\exp\!\big(c\,\ln^2(T/T_{\min})\big),\ \lambda_{\max}\Big],\qquad \lambda_{\min}=5.345,\quad T_{\min}=0.420,\quad c=1.878,\quad \lambda_{\max}=85.
    \label{eq:s-churnlaw}
\end{equation}
 The three parameters were fitted to the energy Wasserstein-1 distance $W_1(\lambda,T)$ between churned and reference ensembles, by minimizing the temperature-averaged $W_1/\sigma_{\mathrm{ref}}$ over $\lambda\in[2,85]$.

\subsection{Free energy, entropy and ensemble averages}
\label{si:freeenergy}
Learning both the forward and backward dynamics turns \name{} into an estimator as well as a sampler: the two SDEs~\eqref{eq:m-fwd}--\eqref{eq:m-bwd} are exact time reversals of one another and share the marginals $p_t$, so the path weights are available in closed form, without the divergence computations required by probability-flow ODE likelihoods \cite{richter2024improved,vargas2024transport,du2026feat}. Discretizing $[0,1]$ into steps $t_0=0<\dots<t_K=1$ turns each SDE into a chain of Gaussian Euler--Maruyama kernels with matched variances. The forward kernel is Eq.~\eqref{eq:m-kernels}, and the backward kernel $\overleftarrow{p}_k(\mathbf{x}_k\mid\mathbf{x}_{k+1})$ is the Gaussian with drift $b_\theta-\tfrac{\sigma^2}{2}s_\theta$ evaluated at $(\mathbf{x}_{k+1},t_{k+1})$ and the same variance $\sigma(t_k)^2\Delta t_k$. A trajectory generated forward from the uniform prior therefore carries the path weight of Eq.~\eqref{eq:m-logw},
\begin{equation}
    \log w(\mathbf{x}_{0:K})
    =-\beta U(\mathbf{x}_K)-\log p_0(\mathbf{x}_0)
    +\sum_{k=0}^{K-1}\log\frac{\overleftarrow{p}_k(\mathbf{x}_k\mid\mathbf{x}_{k+1})}{\overrightarrow{p}_k(\mathbf{x}_{k+1}\mid\mathbf{x}_k)},
    \label{eq:s-logw}
\end{equation}
Written multiplicatively, $w$ is a ratio of two measures on the discrete path $\mathbf{x}_{0:K}$,
\begin{equation}
    w(\mathbf{x}_{0:K})=\frac{\overleftarrow{\mathbb{P}}(\mathbf{x}_{0:K})}{\overrightarrow{\mathbb{P}}(\mathbf{x}_{0:K})},
    \qquad
    \overleftarrow{\mathbb{P}}(\mathbf{x}_{0:K})=e^{-\beta U(\mathbf{x}_K)}\prod_{k=0}^{K-1}\overleftarrow{p}_k(\mathbf{x}_k\mid\mathbf{x}_{k+1}),
    \qquad
    \overrightarrow{\mathbb{P}}(\mathbf{x}_{0:K})=p_0(\mathbf{x}_0)\prod_{k=0}^{K-1}\overrightarrow{p}_k(\mathbf{x}_{k+1}\mid\mathbf{x}_k),
    \label{eq:s-pathratio}
\end{equation}
where $\overrightarrow{\mathbb{P}}$ is the law that actually generates the trajectory---the forward SDE started from the prior---and $\overleftarrow{\mathbb{P}}$ is the backward SDE started from the \emph{unnormalized} Boltzmann weight at $t=1$, whose total mass is $Z$ by construction. Taking the expectation of $w$ under the forward chain therefore integrates $\overleftarrow{\mathbb{P}}$ over all paths, and since each backward kernel is a normalized transition density, marginalizing $\mathbf{x}_0,\mathbf{x}_1,\dots,\mathbf{x}_{K-1}$ in turn removes them one at a time and leaves only the terminal factor,
\begin{equation}
    \mathbb{E}_{\overrightarrow{\mathbb{P}}}\big[w(\mathbf{x}_{0:K})\big]
    =\int\overleftarrow{\mathbb{P}}(\mathbf{x}_{0:K})\,\mathrm{d}\mathbf{x}_{0:K}
    =\int_{\mathbb{T}^D}e^{-\beta U(\mathbf{x}_K)}\,\mathrm{d}\mathbf{x}_K=Z,
    \label{eq:s-unbiased}
\end{equation}
so $w$ is an unbiased importance weight for the partition function. Nothing in this argument uses the quality of the learned drifts: \eqref{eq:s-unbiased} holds for \emph{any} $b_\theta,s_\theta$ and any step size, because the forward and backward kernels are normalized whatever drift they are built from. An imperfect model and a coarse discretization inflate the variance of $w$---and hence the number of trajectories needed---but never bias its mean. Every term in~\eqref{eq:s-logw} is a Gaussian log-density evaluated from the two network heads, plus the terminal energy $U(\mathbf{x}_K)$; the potential enters only here, at estimation time, and is never queried during training.

Algorithm~\ref{alg:infer} states generation and weight accumulation together. The resulting batch of weights supports three estimators. First, the partition function follows from $\widehat Z=\tfrac1M\sum_m w_m$, i.e.\ $\ln\widehat Z=\operatorname{logmeanexp}_m\log w_m$, and hence the \emph{absolute} free energy $F(T)=-k_{\mathrm B}T\ln\widehat Z$; because the model density is referenced to the normalized uniform prior through $\log p_0$, no undetermined additive constant remains. Second, the self-normalized weights $\bar w_m=w_m/\sum_{m'}w_{m'}$ turn any terminal observable into the reweighted average $\langle O\rangle=\sum_m\bar w_m\,O(\mathbf{x}_K^m)$, which stays asymptotically unbiased even when the sampler is imperfect \cite{noe2019boltzmann}; applied to $O=U$ this gives the reweighted mean energy $\langle U\rangle$. 
Third, combining the two yields the entropy
$S(T)=(\langle U\rangle-F)/T
=k_{\mathrm B}\bigl[\beta\langle U\rangle+\ln\widehat Z\bigr]$.
The reliability of all three is monitored by the effective sample size $\mathrm{ESS}=(\sum_m w_m)^2/\sum_m w_m^2$, which shrinks as the weights become heavy-tailed.

\begin{algorithm}[t]
\caption{\name{} inference with forward--backward path weights}
\label{alg:infer}
\begin{algorithmic}[1]
\State \textbf{input:} trained heads $(b_\theta,s_\theta)$; condition $\mathbf{c}$ (e.g.\ temperature $T$, with $\beta=\beta(\mathbf{c})$); noise schedule $\sigma(t)$; grid $0=t_0<\dots<t_K=1$; batch size $M$
\For{$m=1,\dots,M$}
  \State draw $\mathbf{x}_0\sim p_0=\mathrm{Unif}(\mathbb{T}^D)$ (zero-COM);\quad $\log w_m\gets-\log p_0(\mathbf{x}_0)$
  \For{$k=0,\dots,K-1$}
    \State $\Delta t_k\gets t_{k+1}-t_k$;\quad draw $\boldsymbol{\xi}\sim\mathcal{N}(\mathbf{0},I)$ (zero-COM)
    \State $\mathbf{x}_{k+1}\gets\mathrm{wrap}\!\big(\mathbf{x}_k+[b_\theta+\tfrac{\sigma^2}{2}s_\theta](\mathbf{x}_k,t_k,\mathbf{c})\,\Delta t_k+\sigma(t_k)\sqrt{\Delta t_k}\,\boldsymbol{\xi}\big)$ \Comment{forward step, Eq.~\eqref{eq:m-fwd}}
    \State $\log w_m\gets\log w_m+\log\overleftarrow{p}_k(\mathbf{x}_k\mid\mathbf{x}_{k+1})-\log\overrightarrow{p}_k(\mathbf{x}_{k+1}\mid\mathbf{x}_k)$ \Comment{Gaussian kernels}
  \EndFor
  \State $\log w_m\gets\log w_m-\beta\,U(\mathbf{x}_K;\mathbf{c})$ \Comment{only energy evaluation}
  \State store $(\mathbf{x}_K^m,\ \log w_m)$
\EndFor
\State \textbf{free energy:}\quad $\ln\widehat Z=\operatorname{logmeanexp}_m\log w_m$,\quad $F=-k_{\mathrm B}T\ln\widehat Z$
\State \textbf{self-normalized weights:}\quad $\bar w_m=e^{\log w_m}/\sum_{m'}e^{\log w_{m'}}$
\State \textbf{ensemble average / mean energy:}\quad $\langle O\rangle=\sum_m\bar w_m\,O(\mathbf{x}_K^m)$,\quad $\langle U\rangle=\sum_m\bar w_m\,U(\mathbf{x}_K^m)$
\State \textbf{entropy:}\quad $S=(\langle U\rangle-F)/T$
\State \textbf{return} samples $\{\mathbf{x}_K^m\}$, weights $\{\bar w_m\}$, and estimates $(\ln\widehat Z,F,\langle U\rangle,S)$
\end{algorithmic}
\end{algorithm}

For a stiff, low-temperature target the path weights become heavy-tailed and the ESS of this single-temperature estimator collapses, degrading the direct estimate of $F(T)$. Temperature amortization provides a remedy: the same network samples the whole family $\{p_{\mathrm{target}}(\cdot\mid T_k)\}$ along a ladder $T_1>\dots>T_L$, and the multistate Bennett acceptance ratio (MBAR) \cite{shirts2008statistically} reuses every configuration at every rung. With $N_k$ samples drawn at rung $k$ and the reduced potential $u_{kn}=\beta_k U(\mathbf{x}_n)$ of pooled sample $n$, MBAR solves for the dimensionless free energies $f_k$ (gauge $f_0=0$) that minimize the convex objective
\begin{equation}
    \phi(f)=\sum_n\operatorname{logsumexp}_k\!\big(\ln N_k+f_k-u_{kn}\big)-\sum_k N_k f_k,
    \label{eq:s-mbar}
\end{equation}
whose stationarity condition $\sum_n W_{kn}=1$ fixes the per-sample weights and any reweighted average at rung $k$,
\begin{equation}
    W_{kn}=\frac{e^{f_k-u_{kn}}}{\sum_\ell N_\ell\,e^{f_\ell-u_{\ell n}}},
    \qquad
    \langle O\rangle_k=\sum_n W_{kn}\,O(\mathbf{x}_n),
    \qquad
    \langle U\rangle_k=\sum_n W_{kn}\,U(\mathbf{x}_n).
    \label{eq:s-mbarw}
\end{equation}
The $f_k$ determine $\ln Z$ only up to a constant, $\ln\widehat Z(T_k)=-f_k+C$. We pin $C$ by anchoring the ladder to the path-weight estimate at the hottest rung, $C=\ln\widehat Z^{\mathrm{path}}(T_{\mathrm{hot}})+f_{\mathrm{hot}}$, where the liquid is easy, the model marginal is close to $p_{\mathrm{target}}$, and~\eqref{eq:s-logw} is reliable; this single number sets the absolute scale for every colder rung.

These tools give three routes to the \emph{absolute} free energy of a cold target, illustrated in Fig.~\ref{fig:s-freeenergy}. \textbf{(a)}~Bridge the uniform (ideal-gas) reference to the cold end with a dense temperature ladder sampled by classical methods and combine the individual rungs with MBAR. While this gives an unbiased and absolutely normalized estimate through the analytically known ideal-gas end, it requires many finely spaced rungs and is hence computationally expensive. This is the route taken by the reference calculation (Section~\ref{si:km}). \textbf{(b)}~Estimate $F$ at the cold target directly from the forward--backward path weights~\eqref{eq:s-logw}: a single-shot estimate that suffers from high variance, as the ESS collapses at low temperatures. \textbf{(c)}~Combine the two: compute the absolute free energy once at a high temperature from the path weights, where they are reliable, then let the temperature-amortized model generate equilibrium samples at a few colder temperatures and stitch them with a short MBAR ladder to fill in the remaining free energy difference down to the cold target. Route~(c) inherits the absolute anchor of the path-weight estimator and the well-conditioned differences of MBAR, and is the estimator we use for the low-temperature curves in Figure \ref{fig2}.

\begin{figure}[t]
\centering
\begin{tikzpicture}[
  font=\small,
  pw/.style={<->,teal!55!black,line width=1.2pt,shorten >=1.5pt,shorten <=1.5pt},
  gen/.style={-{Stealth[length=4.5pt]},orange!85!black,line width=0.8pt},
  mbar/.style={decorate,decoration={brace,amplitude=4pt,mirror,raise=1pt},black!55,line width=0.6pt},
  tag/.style={font=\footnotesize\bfseries,anchor=east},
  rlab/.style={font=\footnotesize,align=left,anchor=west},
]
  \def\sp{1.05}\def\hw{0.40}
  \def\ya{0}\def\yb{-2.1}\def\ymod{-3.6}\def\yc{-4.8}
  \pgfmathsetmacro{\xend}{7*\sp}
  \def\xr{8.05}
  \draw[->,black!45,line width=0.6pt] (0,1.35) -- (\xend,1.35);
  \node[font=\footnotesize,black!60,anchor=south west] at (0,1.4) {uniform prior $p_0$ ($\beta{=}0$), hot};
  \node[font=\footnotesize,black!60,anchor=south east] at (\xend,1.4) {cold target $p_{\mathrm{target}}$};
  \node[tag] at (-0.65,\ya) {(a)};
  \foreach \i in {0,...,7}{\pgfmathsetmacro{\xx}{\i*\sp}\pgfmathsetmacro{\ff}{\i/7}\enghist{\xx}{\ya}{\ff}}
  \draw[mbar] (-\hw,{\ya-0.42}) -- ({\xend+\hw},{\ya-0.42});
  \node[font=\scriptsize,black!55,anchor=north] at ({\xend/2},{\ya-0.56}) {MBAR};
  \node[rlab] at (\xr,\ya) {PT-MCMC reference:\\swaps $+$ MBAR,\\many rungs (inefficient)};
  \node[tag] at (-0.65,\yb) {(b)};
  \enghist{0}{\yb}{0}
  \enghist{\xend}{\yb}{1}
  \draw[pw] (\hw,\yb) -- ({\xend-\hw},\yb);
  \node[font=\scriptsize,teal!45!black,anchor=south] at ({\xend/2},{\yb+0.05}) {forward--backward SDE (path weights)};
  \node[rlab] at (\xr,\yb) {one-shot path weights:\\high variance at low $T$};
  \node[tag] at (-0.65,\yc) {(c)};
  \def\xanch{3.15}
  \enghist{0}{\yc}{0}
  \enghist{\xanch}{\yc}{0.43}
  \draw[pw] (\hw,\yc) -- ({\xanch-\hw},\yc);
  \node[font=\scriptsize,teal!45!black,anchor=north] at ({\xanch/2},{\yc-0.06}) {path weights};
  \node[font=\scriptsize,anchor=south] at (\xanch,{\yc+0.42}) {anchor: absolute $\ln Z$};
  \def\xmA{4.55}\def\xmB{5.95}\def\xmC{7.35}
  \enghist{\xmA}{\yc}{0.62}\enghist{\xmB}{\yc}{0.81}\enghist{\xmC}{\yc}{1}
  \node[draw=orange!80!black,rounded corners=2pt,fill=orange!8,inner sep=3pt,font=\scriptsize,align=center] (mod) at ({(\xmA+\xmC)/2},\ymod) {amortized model $b_\theta,s_\theta$};
  \draw[gen] (mod) -- ({\xmA},{\yc+0.34});
  \draw[gen] (mod) -- ({\xmB},{\yc+0.34});
  \draw[gen] (mod) -- ({\xmC},{\yc+0.34});
  \draw[mbar] ({\xanch-\hw},{\yc-0.42}) -- ({\xmC+\hw},{\yc-0.42});
  \node[font=\scriptsize,black!55,anchor=north] at ({(\xanch+\xmC)/2},{\yc-0.56}) {short MBAR bridge};
  \node[rlab] at (\xr,\yc) {anchor $+$ short MBAR:\\absolute, well-conditioned};
\end{tikzpicture}
\caption{\textbf{Three routes to the absolute free energy of a cold target.} Each box is a schematic energy histogram at one temperature, from the uniform prior $p_0$ ($\beta{=}0$: broad and red, left) to the cold target $p_{\mathrm{target}}(\cdot\mid T_{\min})$. \textbf{(a)}~The parallel-tempering (PT-MCMC) reference bridges the ideal gas to the cold end along a dense ladder, with replica swaps exchanging configurations between neighboring rungs and MBAR combining them~\eqref{eq:s-mbarw}; unbiased but requiring many rungs. \textbf{(b)}~A single forward--backward SDE path-weight estimate (teal $\leftrightarrow$)~\eqref{eq:s-logw} reaches the cold target directly from the uniform prior, but its variance grows as the effective sample size collapses. \textbf{(c)}~The hybrid estimator: a reliable path-weight estimate anchors the absolute $\ln Z$ at a high temperature, the temperature-amortized model then generates equilibrium samples at a few colder temperatures, and a short MBAR bridge closes the remaining free energy difference.}
\label{fig:s-freeenergy}
\end{figure}
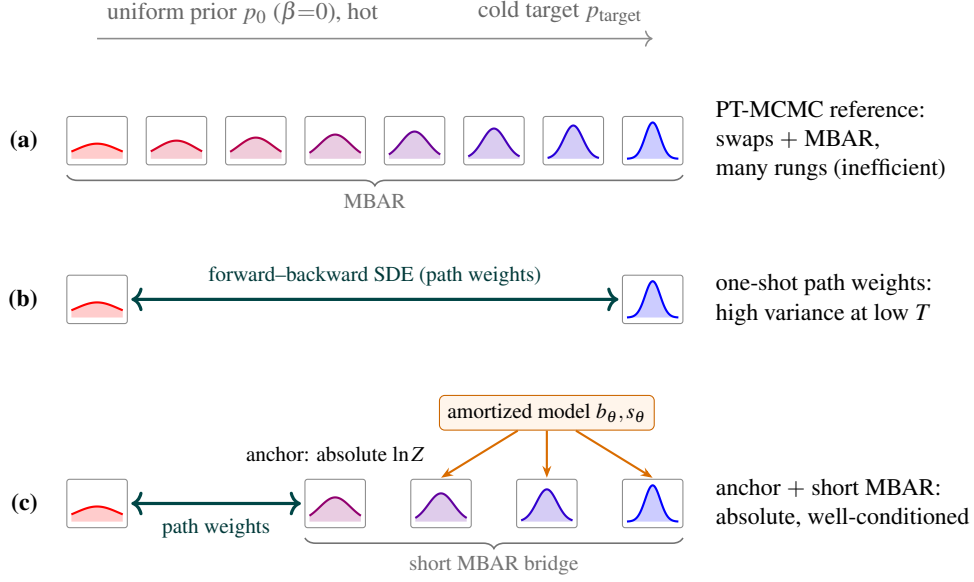

\subsection{Derivation of the inference-time steering weights}
\label{si:rne-derivation}
Let $\mathbb{Q}(\mathbf{x}_{0:K})=p_0(\mathbf{x}_0)\prod_{k}\overrightarrow{q}_k(\mathbf{x}_{k+1}\,|\,\mathbf{x}_k)$ be the law of the proposal chain. For any unnormalized terminal density $\tilde p$ and any normalized backward kernels $\overleftarrow{c}_k$, the weight
\begin{equation}
    w(\mathbf{x}_{0:K})
    =\frac{\tilde p(\mathbf{x}_K)}{p_0(\mathbf{x}_0)}
    \prod_{k=0}^{K-1}\frac{\overleftarrow{c}_k(\mathbf{x}_k\,|\,\mathbf{x}_{k+1})}{\overrightarrow{q}_k(\mathbf{x}_{k+1}\,|\,\mathbf{x}_k)}
    \label{eq:si-generic-w}
\end{equation}
satisfies $\mathbb{E}_{\mathbb{Q}}[w\,O(\mathbf{x}_K)]=\int \tilde p\,O\,\mathrm{d}\mathbf{x}$ for any observable $O$ as the proposal density cancels, and $\mathbf{x}_0,\dots,\mathbf{x}_{K-1}$ integrate out in turn since each $\overleftarrow{c}_k$ is normalized. The choice of $\overleftarrow{c}_k$ therefore affects only the variance. The path weight \eqref{eq:m-logw} is the case $\tilde p=e^{-\beta U}$, $\overrightarrow{q}_k=\overrightarrow{p}_k$, $\overleftarrow{c}_k=\overleftarrow{p}_k$.
 
For steering, let $p_{\theta,k}$ denote the (intractable) marginals of the learned chain, with $p_{\theta,0}=p_0$ and $p_{\theta,K}=p_\theta$. Choosing the tilted model ensemble $\tilde p=p_{\theta}\,e^{\eta r}$ and the exact time reversal of the learned chain, $\overleftarrow{c}_k(\mathbf{x}_k\,|\,\mathbf{x}_{k+1})=\overrightarrow{p}_k(\mathbf{x}_{k+1}\,|\,\mathbf{x}_k)\,p_{\theta,k}(\mathbf{x}_k)/p_{\theta,k+1}(\mathbf{x}_{k+1})$, the marginals telescope against $p_\theta$ and the prior, and \eqref{eq:si-generic-w} reduces to
\begin{equation}
    \log w(\mathbf{x}_{0:K})
    =\eta\, r(\mathbf{x}_K)
    +\sum_{k=0}^{K-1}\log\frac{\overrightarrow{p}_k(\mathbf{x}_{k+1}\,|\,\mathbf{x}_k)}{\overrightarrow{q}_k(\mathbf{x}_{k+1}\,|\,\mathbf{x}_k)},
    \label{eq:si-steering-w}
\end{equation}
free of intractable quantities. Writing $\eta\,r(\mathbf{x}_K)=\sum_{k}\big[r_{t_{k+1}}(\mathbf{x}_{k+1})-r_{t_k}(\mathbf{x}_k)\big]$, which uses only the boundary conditions $r_0\equiv0$ and $r_1=\eta r$, yields the recursion \eqref{eq:m-smc}. The intermediate reward hence cancels from the terminal weight; it matters only through resampling, as the population after step $k$ approximates $\propto p_{\theta,k}\,e^{r_{t_k}}$, so $r_t$ controls how gradually the tilt is introduced. 
 
\subsection{Choices of intermediate reward and proposal for reward tilting}
\label{si:rne}
We steer the pretrained model toward the tilted target $p_\eta\propto p_{\mathrm{target}}\,e^{\eta r}$ by sequential Monte Carlo with the Radon--Nikodym estimator (RNE) \cite{he2025rne}, which propagates $M$ weighted particles through a proposal SDE with drift $a(\mathbf{x},t)$ and an intermediate reward $r_t$, updating each log path weight by Eq.~\eqref{eq:m-smc}; writing $R(\mathbf{x})=\eta\,r(\mathbf{x})$ for the terminal reward, we collect below the choices of $a$ and $r_t$ that this construction leaves open. Note that we are effectively sampling from $p_\eta\propto p_{\theta}\,e^{\eta r}$, $p_{\theta}$ being the terminal distribution of the pretrained model, with $p_{\theta} = p_{\mathrm{target}}$ only when the pretrained model is perfect.

\paragraph{The optimal reward: Doob's $h$-transform.} Zero weight variance is achieved when the intermediate targets are the true noised marginals of the tilt, i.e.
\begin{equation}
    e^{r_t^{\star}(\mathbf{x})}=\mathbb{E}\big[e^{R(\mathbf{x}_1)}\,\big|\,\mathbf{x}_t=\mathbf{x}\big],
    \label{eq:s-doob}
\end{equation}
the Doob's $h$-transform (expectation under the base chain), together with the correspondingly twisted proposal $a$; then the per-step weight is identically one. Equation~\eqref{eq:s-doob} is intractable in general, and the practical rewards below are surrogates for it.

\paragraph{Annealed and Tweedie surrogates, and the Jensen gap.} The \emph{annealed} reward $r_t(\mathbf{x})=\alpha(t)\,R(\mathbf{x})$ scores the noisy state directly whereas the \emph{Tweedie} reward $r_t(\mathbf{x})=R(\hat{\mathbf{x}}_1(\mathbf{x},t))$ scores the predicted clean state~\eqref{eq:s-tweedie}; it replaces the exact $h$-transform $\mathbb{E}[e^{R(\mathbf{x}_1)}\mid\mathbf{x}_t]$ by $e^{R(\mathbb{E}[\mathbf{x}_1\mid\mathbf{x}_t])}$, a delta-method (Jensen) approximation. The Jensen gap is exact at $t\to1$ (the posterior collapses) and for affine $R$, but is large at small $t$ when the conditional law of $\mathbf{x}_1$ is multimodal, since $\hat{\mathbf{x}}_1$ then averages across modes and lands in a low-density region. For a reward that depends on $\mathbf{x}_1$ only through a \emph{linear} collective variable $\xi$, this pathology cannot occur, because $\xi(\hat{\mathbf{x}}_1)=\mathbb{E}[\xi(\mathbf{x}_1)\mid\mathbf{x}_t]$ exactly; only the curvature of $R$ in $\xi$ remains, which the next construction handles in closed form.

\paragraph{Gaussian Doob's $h$-transform for umbrella sampling.} When the collective variable $\xi$ is linear in $\mathbf{x}$ and the reward is Gaussian in $\xi$---as for the umbrella windows of Section~\ref{si:km-umbrella}, $R=\log g$ with $g=\exp[-\tfrac12\sum_j\beta\kappa_j(\xi_j-\bar\xi_j)^2]$---the Doob's $h$-transform~\eqref{eq:s-doob} reduces to a one-dimensional Gaussian integral. Modelling the conditional law of $\xi_j(\mathbf{x}_1)$ given $\mathbf{x}_t$ as Gaussian with mean $\hat\xi_{1,j}=\xi_j(\hat{\mathbf{x}}_1)$ and variance $v_t$, and using $\int\mathcal{N}(z;\mu,v)\,e^{-\frac{b}{2}(z-c)^2}\mathrm{d}z=(1+bv)^{-1/2}\exp[-\tfrac{b(\mu-c)^2}{2(1+bv)}]$, gives the reward
\begin{equation}
    r_t(\mathbf{x})=-\frac12\sum_j\left[\frac{\beta\kappa_j\big(\hat\xi_{1,j}(\mathbf{x},t)-\bar\xi_j\big)^2}{1+\beta\kappa_j v_t}+\log\big(1+\beta\kappa_j v_t\big)\right]-G_0,
    \label{eq:s-gaussh}
\end{equation}
with $G_0$ the (particle-independent) value of the bracket at $t=0$, added back to the evidence.  The posterior variance is set by the empirical conditional second moment $v_t=\mathbb{E}[(\xi(\mathbf{x}_1)-\hat\xi_1(\mathbf{x}_t))^2]$ measured on untilted trajectories.

\subsection{Model architecture and size-transfer}
\label{si:architecture}
Both learned fields---the velocity $b_\theta$ and the score $s_\theta$ of Eq.~\eqref{eq:m-bs}---are emitted as two heads of a single equivariant graph neural network with a shared trunk, acting on the periodic configuration $\mathbf{x}\in\mathbb{T}^D$. We use a torus-adapted PaiNN backbone \cite{schutt2021equivariant}. The same architecture serves every experiment; only the conditioning inputs change between the temperature- and composition-amortized models, and the task-specific hyperparameters, e.g., cutoffs, embedding dimensions, and layer and parameter counts, are collected in Section~\ref{si:arch}.

\paragraph{Equivariances.} The network respects the exact symmetries of the target. \emph{Translation} invariance is built in by forming every pairwise quantity from minimum-image displacements $\mathrm{mic}(\mathbf{x}_i-\mathbf{x}_j)$ under periodic boundaries and by projecting both outputs onto the zero-center-of-mass subspace, matching the prior and the removed global-translation mode. \emph{Rotation and reflection} equivariance follows PaiNN's split into invariant scalar features $\mathbf{s}_i$ and equivariant vector features $\mathbf{v}_i$ per atom: scalars are updated from rotation-invariant distances and inner products, while vectors are updated only by scaling and linear combination of other vectors, so each head outputs a genuine $O(d)$-equivariant tangent vector.

\paragraph{Dynamic graph.} Messages are exchanged on a neighbor graph rebuilt from the current configuration at every network evaluation: an edge $(i,j)$ is present when its minimum-image distance falls inside a finite cutoff window, and each edge distance is expanded in a smooth radial basis modulated by a polynomial envelope that decays to zero with vanishing derivative at the cutoff, so the field stays continuous as pairs enter and leave. Rebuilding the graph dynamically from the live positions rather than freezing a topology from a reference lattice lets the same weights act on the disordered, rearranging configurations that arise while transporting the uniform prior toward the target.

\paragraph{Conditioning.} Thermodynamic and chemical conditions enter as additional per-atom input embeddings, concatenated with the interpolant-time and species embeddings. Temperature enters both through the rescaled force $\beta(T)\mathbf{F}$ that defines the score target and through a log-normalized temperature embedding; composition enters through the species-dependent potential and a composition embedding. Because conditioning only augments the node features, a single backbone amortizes an entire family of Boltzmann ensembles (Section~\ref{si:algorithm}).

\paragraph{Size-transferability.} The locality of the message passing is what makes the learned sampler transfer across system size. Each score target in Eq.~\eqref{eq:m-loss} is the rescaled force $\beta\mathbf{F}(\mathbf{x}_1)$, an intensive per-atom quantity, and each atom's updated field depends only on the neighbors inside its finite cutoff. The two heads therefore define per-atom maps that are well-defined for any atom count $N$: a network trained on moderate supercells can be evaluated, without retraining, on larger cells simply by rebuilding the periodic neighbor graph at the new size. This transfer is accurate only insofar as the distribution of local environments is preserved---the number density and the short- to medium-range correlations must match between training and deployment, and finite-size effects on quantities that depend on correlations longer than the cutoff (or the training cell) are not removed by it.

\subsection{LLM-EA implementation}
\label{si:llmea}

We include the pseudocode for LLM-EA implementation in Algorithm \ref{alg:llm_ea}.

\begin{algorithm}[t]
\caption{Asynchronous LLM-EA active learning with pretrained ATLAS}
\label{alg:llm_ea}
\begin{algorithmic}[1]
\Require Pretrained ATLAS model $p_\theta$, composition space, reward/objective function $r$, an LLM

\State Sample $\{\textbf{c}_i\}$ uniformly and initialize composition buffer $\mathcal{C}$
\State Initialize Pareto frontier buffer $\mathcal{B} = \varnothing$

\While{evaluation budget is not exhausted}
\If{computing slot is available}
    \State Submit the ATLAS evaluation jobs to evaluate $r(p_\theta(\cdot, \mathbf{c_i}))$ from $\mathcal{C}$
    \State Collect results $\{r_i\}$, update $\mathcal{B}$
    \State Prompt the LLM to refill $\mathcal{C}$ from the accumulated knowledge in $\mathcal{B}$
\Else
    \State sleep and wait for the next available slot
\EndIf
\EndWhile

\State \Return the evaluated population in $\mathcal{B}$
\end{algorithmic}
\end{algorithm}

\section{Experimental setup}
\subsection{Additional details on validation of ATLAS}
\label{si:km}

\paragraph{Kob–Andersen reference data.}
Because conventional Monte Carlo does not equilibrate the supercooled KA model at a single temperature, ground-truth ensembles are generated by an instance of PT-MCMC called non-reversible parallel tempering (NRPT) \cite{hukushima1996exchange,syed2022non}, which couples a ladder of $K$ replicas that anneal a hot, freely mixing liquid to the cold target. NRPT combines three ingredients: local exploration by Hamiltonian Monte Carlo \cite{duane1987hybrid,neal2011mcmc}, which traverses the stiff neighbor-pair directions that overdamped Langevin cannot at the cold end; replica communication by the deterministic even/odd non-reversible swap scheme, which gives ballistic rather than diffusive transport along the ladder; and adaptive schedule optimization, which repositions the rungs during warm-up so that the cumulative swap rejection is uniform across the ladder \cite{syed2022non}. After warm-up the step sizes and schedule are frozen, so stored configurations are unbiased draws from the tempered densities. For $N=64$ we ran $8$ independent chains of $10{,}000$ samples per rung, merged by rung to $80{,}000$ configurations per temperature over a $16$-rung ladder spanning $T\in[0.2,5.0]$; generated ensembles are compared against the reference rung nearest each reported temperature ($T=0.5,0.3,0.2$ against rungs at $0.507,0.288,0.200$).

One NRPT sweep advances all $K=16$ rungs by $n_{\mathrm{ex}}=20$ HMC moves (each $L=10$ leapfrog steps, i.e.\ $L{+}1=11$ force evaluations) and then attempts one deterministic even/odd swap ($\approx2$ energy evaluations per rung), for a per-sweep cost of $K\,[\,n_{\mathrm{ex}}(L{+}1)+2\,]=16\times(20{\cdot}11+2)=3552$ target evaluations. The adaptive warm-up runs $6000$ sweeps, $\approx2.13\times10^{7}$ evaluations per chain, over which the rung schedule and the per-rung step sizes are tuned while the chain relaxes from a random initial configuration; both are then frozen, so the collected draws are unbiased. Collection runs a further $M=10{,}000$ sweeps, $\approx3.55\times10^{7}$ evaluations per chain, storing one configuration per rung per sweep. Successive stored configurations are therefore separated by $n_{\mathrm{ex}}=20$ HMC moves---the intrinsic thinning that decorrelates the stiff cold rung---with no additional subsampling. Each chain thus costs $\approx5.68\times10^{7}$ target evaluations and the full eight-chain reference $\approx4.55\times10^{8}$.
The reference values in Fig.~\ref{fig2}e are generated from these samples. We solve the MBAR equations for $\ln Z(T{=}0.2)$, and plot the relative free energy error $|\Delta F|/|F|=|\ln Z-\ln Z_{\mathrm{ref}}|/|\ln Z_{\mathrm{ref}}|$ against the fully converged value $\ln Z_{\mathrm{ref}}$ (all eight chains at $M=10{,}000$). The `matched' value uses $M=2000$ configurations, exactly the per-temperature sample count of \name{}, whose cost is the training evaluations ($\approx4\times10^{5}$) plus one generation of $2000$ samples per rung that MBAR then reuses. The `long-run' value uses $M=5000$. Lastly, using $M=10{,}000$ gives zero error by construction and is omitted from the logarithmic axis.

\paragraph{Cage-escape free energy profile.}
\label{si:km-umbrella}
The steering demonstration (Fig.~\ref{fig2}d) computes the potential of mean force for pulling a tagged $B$ particle out of its first coordination shell, along a collective variable (CV) built from a tagged $B$--$A$ pair. Two symmetries of the KA ensemble constrain the CV: translation invariance forces it to depend on relative positions only, and permutation invariance within a species means a single tagged particle carries no structural anchor---so a well-defined single-particle pull coordinate requires a second tagged particle. We therefore tag one $B$ particle $b$ and one $A$ particle $a$ and use their minimum-image displacement $\boldsymbol{\Delta}=\mathrm{mic}(\mathbf{x}_b-\mathbf{x}_a)$. Its length $|\boldsymbol{\Delta}|$ is the physically clean radial coordinate (rotation invariant, no auxiliary restraint), whose profile is exactly the pair potential of mean force $-k_{\mathrm B}T\log[2\pi r\,g_{AB}(r)/L^2]$. The runs reported here instead use the \emph{axial} coordinate $\xi=\Delta_y$ together with a harmonic transverse restraint $u_\perp=\tfrac12\kappa_\perp\Delta_x^2$ (stiffness $\kappa_\perp=100$) that keeps the pull on-axis; this restraint is part of the ensemble whose free energy is reported and is applied identically in the reference. The restraint opens a computable escape channel: below the separation $r^\star$ solving $V_{AB}'(r^\star)+\kappa_\perp r^\star=0$ the pair slides sideways at fixed separation rather than compressing, so only $\xi\gtrsim r^\star$ measures compression.

The profile is stratified into overlapping umbrella windows with harmonic biases $u_k=\tfrac12\kappa_k(\xi-\xi_k)^2$. Each window is a tilted target $p_{\mathrm{target}}\,g_k$ with $g_k=\exp[-\beta(u_\perp+u_k)]$, sampled by the inference-time RNE scheme (Section~\ref{si:rne}) with the closed-form Gaussian $h$-transform reward~\eqref{eq:s-gaussh} and score-space guidance, resampling when the effective sample size falls below half the population. We report the experiment at the target temperature $T=0.2$, which stiffens the umbrella and lowers ancestral diversity to a few percent of the population; since the RNE evidence is unbiased at any effective sample size, we reduce variance by pooling $K=6$ independent seeds at $M=1024$ particles and $700$ integration steps, whose seed-to-seed scatter also furnishes an honest error bar.

\subsection{Parameters and training details}
\label{si:training}
This section details the hyperparameters and design choices behind all experiments reported in the paper. Two families of models are used throughout: a single temperature-amortized sampler for the two-dimensional KA glass former (Section~\ref{si:km}), and composition-amortized samplers for the metallic glasses---Cu-Zr and Cr-Co-Ni with EAM potentials, the five-element Fe-Ni-Cr-Co-Cu glass with EAM, and the eight-element Fe-Mn-Ni-Ti-Cu-Cr-Co-Al pool with the MACE-MPA-0 MLIP. 

\begin{table*}[t]
\centering
\small
\renewcommand{\arraystretch}{1.15}
\begin{tabular}{@{}p{0.19\textwidth}p{0.36\textwidth}p{0.36\textwidth}@{}}
\hline
 & \textbf{KA (temperature-amortized)} & \textbf{Metallic glasses (composition-amortized)} \\
\hline
\textbf{System specific} & & \\
System & $80{:}20$ KA mixture in $d=2$; $N=64$, $\rho=1.0$, $L=8$; $n_A=51$, $n_B=13$ & Cu-Zr, Cr-Co-Ni, Fe-Ni-Cr-Co-Cu, Fe-Mn-Ni-Ti-Cu-Cr-Co-Al in $d=3$; $N=128$ \\
Potential & truncated--shifted Lennard--Jones, $C^1$ logarithmic core below $0.8\,\sigma_{AA}$ & EAM \cite{hale2018evaluating}; MACE-MPA-0 \cite{batatia2025foundation} for the eight-element pool; linear core below $1.5$\,\AA \\
Cell & square, side fixed by $\rho=1.0$ & cubic, side from the composition mixing rule~\eqref{eq:densrule} \\
Conditioning & temperature, $T\in[0.2,5.0]$, sampled log-uniformly & composition vector $\mathbf{c}$, drawn per graph, at fixed $T$ ($700$\,K for Cu-Zr, $500$\,K otherwise) \\
\hline
\textbf{Target and schedules} & & \\
Interpolant & deterministic, $\gamma(t)\equiv0$ & same \\
Noise schedule $\sigma(t)$ & variance-exploding, $\sigma_{\min}{=}0.01$, $\sigma_{\max}{=}1.0$ & EDM \cite{karras2022elucidating}, $\sigma_{\min}{=}0.02$, $\sigma_{\max}{=}0.8$, $\rho_{\mathrm{EDM}}{=}7$ \\
$\alpha(t)$ & $\kappa(t)/\kappa(1)$, $\kappa(t){=}\int_0^t\!\sigma(u)^2\,\mathrm{d}u$ & same \\
Prior $p_0$ & uniform on the cell, zero-COM & same \\
SDE steps (generation) & $250$ (Euler--Maruyama) & $100$ (Euler--Maruyama) \\
\hline
\textbf{Network} & & \\
Backbone & torus PaiNN, $6$ layers, width $64$ & same \\
Radial basis / cutoff & $10$ Gaussians, $[0.4,2.6]$ in LJ units, poly.\ envelope & $10$ Gaussians, cutoff $5$\,\AA, poly.\ envelope \\
Graph & dynamic (rebuilt per eval), minimum image & same \\
Time / species embed.\ & $10$ / $8$ dims & same \\
Condition embed.\ & $16$-dim log-normalized temperature & composition fractions $\mathbf{c}$, one entry per element \\
Trainable parameters & $334{,}164$ ($290{,}640$ trunk $+\,2{\times}21{,}762$) & same trunk and heads; input layer scales with the number of elements \\
\hline
\textbf{Optimization} & & \\
Optimizer & Adam, lr $3{\times}10^{-4}$, grad-norm clip $1.0$ & same \\
Outer $\times$ inner steps & $200\times1000$ ($2{\times}10^5$ updates) & $400\times200$ (Cu-Zr, Cr-Co-Ni) or $300\times200$ (high-entropy glasses) \\
Batch size & $256$ (generation and regression) & same \\
Replay buffer & size $5000$, $20\%$ refreshed per outer step & size $4000$, fully refreshed per outer step \\
Score-target norm clip & $750$ & same \\
Min-distance filter & $0.01$ (LJ units) & --- \\
Force evaluations & $\approx2.0\times10^5$ in total & $\approx1.6\times10^6$ or $\approx1.2\times10^6$ in total ($400$ or $300$ outer steps) \\
Weighting via \eqref{eq:iw-loss} & No & Yes \\
Target fraction & --- & $\kappa=0.8$ \\
\hline
\end{tabular}
\caption{Hyperparameters and design choices for the two families of \name{} models: the temperature-amortized KA sampler ($N=64$, $\rho=1.0$) and the composition-amortized metallic-glass samplers. Entries reading ``same'' repeat the KA value; ``---'' marks a setting that does not apply.}
\label{si:tab-hparams}
\end{table*}

\paragraph{Deterministic interpolant.} We set the bridge width to zero, $\gamma(t)\equiv0$, so the interpolant reduces to the pure-transport bridge $\mathbf{x}_t=\mathrm{wrap}(\mathbf{x}_0+\alpha(t)\boldsymbol{\Delta})$ and the velocity target drops its $\dot\gamma(t)\gamma(t)\hat{\mathbf{s}}$ correction, leaving $\hat{\mathbf{b}}=\dot\alpha(t)\boldsymbol{\Delta}$ in Eq.~\eqref{eq:m-loss}. The score head is unaffected by this choice and remains available for the SDE sampler, the churn corrector and inference-time steering---so only the injected bridge noise is removed.

\paragraph{$\beta$-factored score head.} For the temperature-amortized KA model the score head does not regress $\hat{\mathbf{s}}=\beta\mathbf{F}(\mathbf{x}_1)$ directly. It regresses the temperature-independent force, $f_\theta(\mathbf{x}_t,t,T)\approx\mathbb{E}[\mathbf{F}(\mathbf{x}_1)\mid\mathbf{x}_t]$, and the score is formed analytically as $s_\theta=\beta f_\theta$ wherever it is used. Over $T\in[0.2,5]$ the factor $\beta$ varies $25$-fold, so regressing $\beta\mathbf{F}$ would weight the loss towards the cold end by $\beta^2$; factoring it out gives a target of uniform scale across temperatures. The score-target clip ($750$) is applied to $\beta\mathbf{F}$ before the factor is removed.

\paragraph{Schedules.} Both the transport schedule $\alpha(t)$ and the sampling noise $\sigma(t)$ derive from a single noise schedule. For the KA model this is the geometric variance-exploding schedule with $\sigma_{\min}=0.01$ and $\sigma_{\max}=1.0$: $\sigma(t)=\sigma_{\min}(\sigma_{\max}/\sigma_{\min})^{1-t}\sqrt{2\ln(\sigma_{\max}/\sigma_{\min})}$, and $\alpha(t)=\kappa(t)/\kappa(1)$ is the normalized cumulative diffusion variance $\kappa(t)=\int_0^t\sigma(u)^2\,\mathrm{d}u$, which rises monotonically from $\alpha(0)=0$ to $\alpha(1)=1$. The metallic-glass models keep the same $\alpha(t)=\kappa(t)/\kappa(1)$ construction but replace the geometric interpolation of $\sigma$ by the EDM schedule of Karras et al.~\cite{karras2022elucidating}, which interpolates $\sigma^{1/\rho_{\mathrm{EDM}}}$ linearly in $t$ between $\sigma_{\max}=0.8$ at $t=0$ and $\sigma_{\min}=0.02$ at $t=1$, with exponent $\rho_{\mathrm{EDM}}=7$. The exponent concentrates the noise levels towards the low-noise end, where the target is stiffest, which is what allows generation with $100$ integration steps rather than the $250$ used for KA.

\paragraph{Conditioning and amortization.} The whole temperature family is learned by one network: each batch draws its temperature log-uniformly on $[0.2,5.0]$ and $T$ enters the training signal only through the rescaled force $\beta(T)\mathbf{F}$ and the log-normalized temperature embedding (Section~\ref{si:arch}). Sampling and embedding $T$ on a logarithmic axis prevents the low-temperature rungs---where the target is stiffest and matters most---from being starved relative to the broad high-temperature range. The composition fraction is uniformly sampled over the simplex.

\paragraph{Replay buffer and force evaluations.} For the KA model the buffer holds $5000$ labeled configurations; it is filled once at the start and then $20\%$ ($1000$) are regenerated and re-labeled each outer step, so training consumes only $\approx 2.0\times10^{5}$ force evaluations in total (each a single $N=64$ force call), amortized across the $2\times10^{5}$ gradient updates---roughly $2.5\times10^{2}$ regression samples per labeled configuration. Crucially, only forces are queried during training; the potential energy $U$ itself is never evaluated until estimation time (Eq.~\eqref{eq:m-logw}). 

The metallic-glass models use a smaller buffer of $4000$ configurations but refresh it completely at every outer step, so each outer step regenerates and re-labels the whole buffer before taking its $200$ inner gradient steps. Over the $400$ outer steps used for Cu-Zr and Cr-Co-Ni this amounts to $\approx1.6\times10^{6}$ force evaluations for $8\times10^{4}$ updates, and over the $300$ outer steps used for the high-entropy glasses to $\approx1.2\times10^{6}$ force evaluations for $6\times10^{4}$ updates.

\paragraph{Inference-time sampling for the reported observables.} The structural and energetic KA panels (radial distribution functions and energy histograms) and all samples entering the free energy and entropy estimates are drawn with the SNR-matched churn sampler of Section~\ref{si:pc}, which adds one score-paired Langevin kick per Euler step of the $250$-step generation SDE at a single extra network evaluation and no additional force evaluation. The per-step churn budget is $h(t)=\lambda(t)\,\min\!\big(2(r\,\|\mathbf{z}\|/\|s_\theta\|)^2,\,h_{\max}\big)$ with signal-to-noise ratio $r=0.16$ and cap $h_{\max}=0.05$. Its strength follows the analytic law~\eqref{eq:s-churnlaw}, together with the score gain $g(T)=\max(1,1/T)$ and the size-scaled score clip $750\sqrt{N/64}$ (Section~\ref{si:pc}). Because the kick uses the learned score rather than $\nabla U$, churn costs zero force or energy calls.

\subsection{Architecture details}
\label{si:arch}
We record here the concrete choices that specialize the generic architecture of Section~\ref{si:architecture} to the two-dimensional KA glass; the composition-amortized metallic-glass models follow the same two-head template with the modifications described at the end of this section. The trunk is a torus PaiNN backbone \cite{schutt2021equivariant} of $6$ message/update blocks at hidden width $64$, each followed by an equivariant RMS normalization of the scalar and vector channels. The dynamic neighbor graph keeps an edge $(i,j)$ when its minimum-image distance lies in $[r_{\min},r_{\max}]=[0.4,2.6]$ in LJ units, where $r_{\max}=2.6$ coincides with the potential cutoff $2.5\,\sigma_{AA}$ plus a small skin and $r_{\min}=0.4$ excludes only near-coincident pairs; each edge distance is expanded in $10$ Gaussian radial basis functions modulated by a polynomial envelope (exponent $5$). Each atom is initialized by concatenating a $10$-dimensional embedding of the interpolant time $t$, a learned $8$-dimensional embedding of the fixed species ($A$ or $B$), and a $16$-dimensional sinusoidal embedding of the conditioning temperature, normalized as $(\ln T-\ln T_{\min})/(\ln T_{\max}-\ln T_{\min})$ over $[T_{\min},T_{\max}]=[0.2,5.0]$ so that the cold rungs occupy a fair fraction of the input range (paired with the log-uniform temperature sampling used in training). The resulting $34$-dimensional node features are linearly mapped to width $64$ and passed through the $6$ blocks, after which two independent gated-equivariant vector heads---each a $64\to32\to1$ reduction over the vector channels---emit the velocity $b_\theta$ and the score $s_\theta$, with the center-of-mass component of each output removed. Sharing the entire backbone, the two heads bring the model to $334{,}164$ trainable parameters: $290{,}640$ in the shared trunk and $21{,}762$ in each head.

The composition-amortized models keep the trunk and the two heads unchanged and differ only in the neighbor graph and the node inputs. The dynamic graph uses a cutoff of $5$\,\AA{} for all metallic glasses---Cu-Zr, Cr-Co-Ni, Fe-Ni-Cr-Co-Cu and the eight-element Fe-Mn-Ni-Ti-Cu-Cr-Co-Al pool---in place of the LJ-unit window $[0.4,2.6]$ of the KA model. Each atom is initialized from the same time and species embeddings, but the temperature embedding is replaced by the composition-fraction vector $\mathbf{c}$ itself, whose length equals the number of elements in the pool, so the input width---and with it the first linear layer---grows by one channel per element. The conditioning is attached per graph rather than per model: every configuration in a batch may carry a different composition, so one gradient step covers many points of composition space at once, which is what makes composition amortization cheap to train. Temperature is held fixed within each of these models ($700$\,K for Cu-Zr, $500$\,K for the remaining metallic glasses), so no temperature embedding is needed and $\beta$ is a constant in the score target $\hat{\mathbf{s}}=\beta\mathbf{F}$.

\newpage
\section{Representative amorphous structures}
Figs.\,\ref{fig_fe5} and \ref{fig_fe8} show representative ATLAS-generated structures selected from the Pareto fronts reported in Fig.\,\ref{fig4}c,d of the main text. The selected compositions span the trade-off between shear modulus $G$ and Pugh ratio $B/G$, from stiffer alloys with larger $G$ to more ductile candidates with larger $B/G$. For each composition, two independently generated configurations are shown to illustrate the structural diversity of the sampled amorphous ensemble. The corresponding composition, bulk modulus $B$, shear modulus $G$ and Pugh ratio $B/G$ are reported above each pair.

\begin{figure}[!h]
  \centering
  \includegraphics[width=\textwidth]{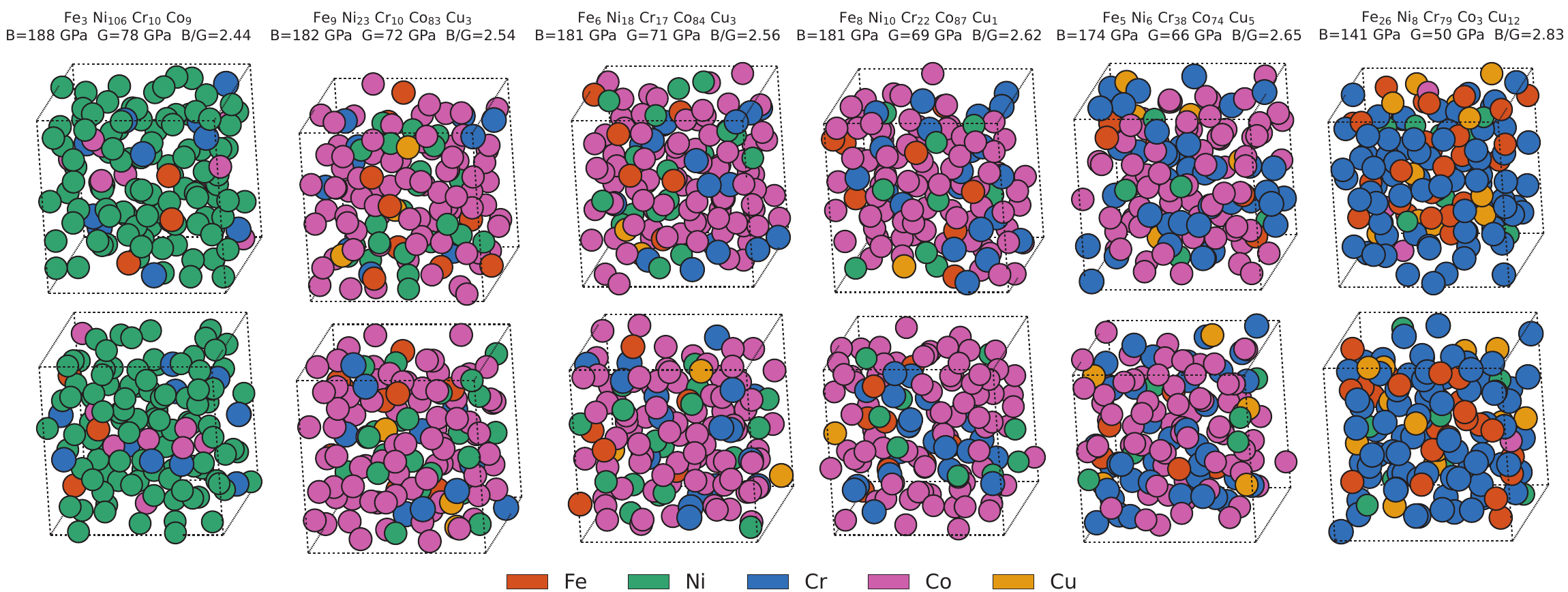}
  \caption{\textbf{Representative Pareto-optimal structures in the five-element metallic glass space.}
  Independently generated amorphous configurations selected along the Pareto front of the Fe-Ni-Cr-Co-Cu search in Fig.\,\ref{fig4}c. Each column shows two ATLAS samples at the same composition. The composition and corresponding ensemble-averaged bulk modulus $B$, shear modulus $G$ and Pugh ratio $B/G$ are indicated above each pair.}
  \label{fig_fe5}
\end{figure}

\begin{figure}[!h]
  \centering
  \includegraphics[width=\textwidth]{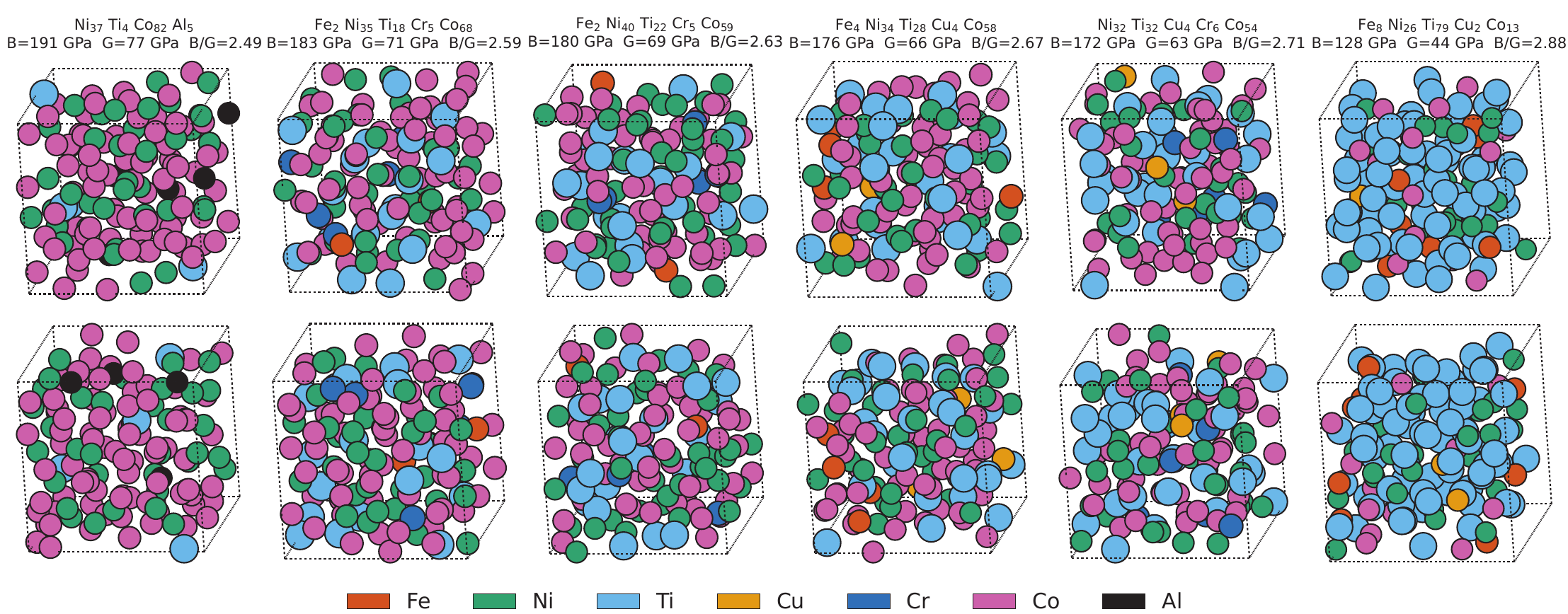}
  \caption{\textbf{Representative Pareto-optimal structures in the eight-element metallic glass space.}
  Independently generated amorphous configurations selected along the Pareto front obtained by choosing up to five elements from the Fe-Mn-Ni-Ti-Cu-Cr-Co-Al pool in Fig.\,\ref{fig4}d. Each column shows two ATLAS samples at the same composition, with the composition and corresponding ensemble-averaged $B$, $G$ and $B/G$ reported above.}
  \label{fig_fe8}
\end{figure}

\bibliography{refs.bib}
\end{document}